\documentclass[a4paper,USenglish,cleveref,nameinlink,autoref,thm-restate]{lipics-v2021}

\usepackage{caption}
\usepackage{subcaption}
\usepackage[dvipsnames]{xcolor}
\usepackage[backgroundcolor=white,linecolor=black,bordercolor=black,textcolor=black,textsize=small]{todonotes}
\usepackage{mathtools}
\usepackage{xspace}
\usepackage{csquotes}
\usepackage[export]{adjustbox}
\usepackage{apptools}

\usepackage[appendix=append,bibliography=common]{apxproof}

\makeatletter
\newcommand{\whenappendix}[1]{%
  \ifthenelse{\equal{\axp@appendix}{append}}%
    {#1}{}%
}

\ifthenelse{\equal{\axp@appendix}{append}}{\newcommand{\restateref}[1]{\IfAppendix{\textsfbfrm{(\hyperref[#1]{$\star$})}}{\textsfbfrm{(\hyperref[#1*]{$\star$})}}}}{\newcommand{\restateref}[1]{\relax}}

\makeatother

\newcommand{\textsfbf}[1]{{\textsf{\textbf{\normalsize{#1}}}}}
\newcommand{\textsfbfrm}[1]{{\textsf{\textbf{{\em\normalsize{#1}}}}}}

\newcommand{\cont}{\mathrm{cont}}
\newcommand{\Econt}{\ensuremath{E_{\mathrm{cont}}}\xspace}

\newcommand{\ECcont}{\ensuremath{E_{\mathrm{cont}}^+}\xspace}

\newcommand{\rev}{\mathrm{rev}}
\newcommand{\mirr}{\mathrm{mirr}}
\newcommand{\I}{\ensuremath{\mathcal{I}}\xspace}
\newcommand{\rep}{\mathrm{rep}}

\newcommand{\Arcs}{\ensuremath{\vec E}}
\newcommand{\arcs}{\mathcal A}

\newcommand{\augH}{{H^+}}
\newcommand{\precaugH}{\prec_H^+}
\newcommand{\succaugH}{\succ_H^+}
\newcommand{\precoverlineaugH}{\prec_{\overline{H}}^+}

\nolinenumbers

\newcommand{\arxiv}[1]{#1} 
\newcommand{\lipics}[1]{} 

\let\emph\relax\DeclareTextFontCommand{\emph}{\color{NavyBlue}\em}

\title{Towards the Recognition of~Oriented~Interval~Graphs}

\author{Lukas P.\ Bachmann}{Universität Passau, Passau, Germany}{bachmanp@fim.uni-passau.de}{https://orcid.org/0009-0003-3749-6265}{}
\author{Jiří Fiala}{Charles University, Prague, Czech Republic}{fiala@kam.mff.cuni.cz}{https://orcid.org/0000-0002-8108-567X}{Supported by grant no. 25-16847S of the Czech Science Foundation (GAČR).}
\author{Miriam Münch}{Universität Passau, Passau, Germany}{miriam.muench@uni-passau.de}{https://orcid.org/0000-0002-6997-8774}{}
\author{Ignaz Rutter}{Universität Passau, Passau, Germany}{rutter@fim.uni-passau.de}{https://orcid.org/0000-0002-3794-4406}{}
\author{Peter Stumpf}{Charles University, Prague, Czech Republic}{stumpf@kam.mff.cuni.cz}{https://orcid.org/0000-0003-0531-9769}{}
\author{Alexander Wolff}{Universität Würzburg, Würzburg, Germany \and \url{https://www.informatik.uni-wuerzburg.de/en/algo/team/wolff-alexander}}{}{https://orcid.org/0000-0001-5872-718X}{}

\authorrunning{L.~P.~Bachmann, J.~Fiala, M.~Münch, I.~Rutter, P.~Stumpf,
  and A.~Wolff}

\Copyright{Lukas P.\ Bachmann, Jiří Fiala, Miriam Münch, Ignaz Rutter,
  Peter Stumpf, and Alexander Wolff}



\arxiv{\hideLIPIcs}

\lipics{
\relatedversion{}
\relatedversiondetails{Full Version}{https://arxiv.org/abs/xxxx.yyyyy}
}

\EventEditors{Philip Bille, Seth Pettie, and Sabine Storandt}
\EventNoEds{2}
\EventLongTitle{34th European Symposium on Algorithms (ESA 2026)}
\EventShortTitle{ESA 2026}
\EventAcronym{ESA}
\EventYear{2026}
\EventDate{August 31 to September 4, 2026}
\EventLocation{L'Aquila, Italy}
\EventLogo{}
\SeriesVolume{388}
\ArticleNo{164}

\begin{document}

\maketitle

\begin{abstract}
  \emph{Oriented interval graphs}, a recent generalization of interval graphs introduced by Gutowski et al.~[GD 2022], are intersection graphs of intervals, each of which is oriented either left or right.
  Such a representation defines a \emph{mixed} intersection graph: overlapping intervals with the same orientation define a (directed) \emph{arc}; nested intervals (irrespective of the orientations of the intervals) and overlapping intervals of opposite orientations define an (undirected) \emph{edge}.
  An oriented interval representation of a mixed graph $G$ can be described combinatorially by the combination of (i)~an \emph{orientation} $\varphi \colon V(G) \to \{-1,1\}$ of all intervals, (ii)~a \emph{clique ordering} $\sigma$, and (iii)~a set~$\Econt \subseteq E(G)$ of \emph{containment edges}, which are represented by nested intervals.  The non-trivial dependencies between these three ingredients make the recognition of oriented interval graphs a challenging problem.

  In this paper, we take steps towards a general recognition algorithm by studying how orientation, clique ordering, and containment edges influence and restrict each other.  We characterize the orientations that are consistent with a given set of containment edges as well as the clique orderings that are consistent with a given orientation.  Based on these characterizations, we give linear-time algorithms for two constrained versions of the recognition problem where, in addition to the mixed input graph~$G$, either the set of containment edges~$\Econt$ or the orientation~$\varphi$ is prescribed.  This improves a quadratic-time algorithm of Gutowski et al.\ for the case that all vertices have the same orientation; an assumption that determines both the orientation and the containment edges.
  In particular, this also solves the recognition problem for oriented proper (or unit) interval graphs.
\end{abstract}

\renewcommand{\topfraction}{.9}
\renewcommand{\bottomfraction}{.9}
\renewcommand{\textfraction}{.1}

\section{Introduction}
\label{sec:intro}

A \emph{mixed graph} is a graph that may contain both (undirected) edges and (directed) \emph{arcs}.
Mixed graphs were introduced by Sotskov and Tanaev~\cite{sotskov1976chromatic} and reintroduced by Hansen et al.~\cite{HansenKW97} and have been studied in the context of scheduling problems~\cite{sotskov2020mixed}, (quasi-)upward planar drawings~\cite{BinucciD16,BinucciDP14,FratiKPTW14} and extensions of partial orientations~\cite{Bang-JensenHZ18}.

An~\emph{intersection representation} of a graph~$G$ is a
mapping~$\rho \colon V(G) \to \mathcal C$ from the vertex set~$V(G)$ of~$G$, 
to a class~$\mathcal{C}$ of geometric objects such that $\{u,v\}$ is an
edge of~$G$ if and only if~$\rho(u) \cap \rho(v) \neq \emptyset$.
Every class~$\mathcal{C}$ of geometric objects gives rise to a class
of graphs that admit an intersection representation with objects
in~$\mathcal C$.  Prominent examples include interval graphs
(intersection graphs of intervals on the real line), circular-arc
graphs (intersection graphs of arcs on a circle), circle graphs
(intersection graphs of chords of a circle), and permutation graphs
(intersection graphs of line segments connecting two parallel lines).

Interval graphs form a well-understood class in structural and
algorithmic graph theory, with applications spanning from scheduling
problems to genome analysis~\cite{golumbic2004algorithmic}.  Using
so-called PQ-trees, one can test in linear time whether a graph~$G$ is
an interval graph and, if yes, compute an interval representation
of~$G$~\cite{BoothL1976}.  Some problems, such as coloring or maximum
clique, that are NP-hard in general can be solved in linear time for
interval graphs~\cite{GareyJMP80,Hsu85}.

Gutowski, Mittelstädt, Rutter, Spoerhase, Wolff, and Zink~\cite{gmrswz-cmdig-GD22} introduced a natural extension of interval representations where each interval is assigned an orientation, namely left or right (assuming the intervals are placed horizontally on the real line).
This \emph{oriented interval representation} yields a mixed intersection graph: overlapping intervals that have the same orientation define an arc pointing from the earlier to the later interval according to their orientation; nested intervals (regardless of their respective orientations) and overlapping intervals of opposite orientations define an edge; see~\cref{fig:example-1}.

    \begin{figure}[b]
        \centering
        \begin{subcaptiongroup}
        \textsfbf{(a)} \includegraphics[page=1,valign=t]{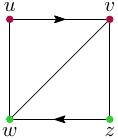}
        \phantomcaption\label{fig:example1-a} \hfill
        \textsfbf{(b)} \includegraphics[page=2,valign=t]{fig/intro-example}
        \phantomcaption\label{fig:example1-b}\hfill
        \textsfbf{(c)} \includegraphics[page=3,valign=t]{fig/intro-example}
        \phantomcaption\label{fig:example1-c}
        \end{subcaptiongroup}
        \caption{(a) A mixed graph $G$.
        (b,c) $\varphi$-oriented interval reps.\ of~$G$, representing~$\{v,w\}$ differently.}
        \label{fig:example-1}
    \end{figure}
A natural question is whether these \emph{oriented interval graphs} can be
recognized efficiently.  Even though directed and mixed graphs arise naturally in many applications, they have mostly been ignored in the study of intersection representations.
Continuing the study of Gutowski et al., we show that even for very simple geometric objects, namely oriented intervals, the recognition problem for oriented intersection graphs poses significant challenges.

Classical interval graphs can be combinatorially described by a \emph{clique ordering}, that is, a graph is an interval graph if and only if its maximal cliques can be ordered such that for each vertex $v$, all cliques that contain $v$ appear consecutively~\cite{FulkersonGross1965}.
For oriented interval graphs, however, the situation is more involved, as an oriented interval representation of a mixed graph $G$ can be described combinatorially by the combination of (i)~an \emph{orientation} $\varphi \colon V(G) \to \{-1,1\}$ of all intervals, (ii)~a clique ordering $\sigma$ of the underlying undirected graph, and (iii)~a subset~$\Econt \subseteq E(G)$ of the edges that are represented by~\emph{containment}, that is, nested intervals.  These three elements are, however, not independent but influence and restrict each other in a non-trivial fashion.
For example, if two vertices~$u$ and~$v$ have the same orientation, say to the right, but they are connected by an undirected edge, then this edge must be represented by containment;
known containments may restrict the choices of the clique ordering, and clique orderings may decide certain containments and overlaps, which influences the orientations.

For a mixed graph $G$, the recognition problem boils down to finding a triplet~$(\varphi, \sigma, \Econt)$ whose elements are \emph{consistent} with each other in the sense that together they describe an oriented interval representation of $G$. 
So far, the recognition problem has only been solved in the case where all intervals are oriented to the right, which, in addition to the orientation~$\varphi$, uniquely determines the containment edges~\Econt.
Gutowkski et al.~\cite{gmrswz-cmdig-GD22} showed that, in this setting, a clique ordering that is consistent with the orientation (and the containment edges) can be computed in time quadratic in the number of graph vertices.

\subparagraph{Our Contribution.}

In this paper, we take steps towards the general recognition problem for oriented interval graphs by studying the interconnection between $\sigma$, $\varphi$ and~\Econt.
We call a triplet~$(\varphi, \sigma,\Econt)$ \emph{consistent} if there is a corresponding oriented interval representation, and we call any pair of two such elements \emph{consistent} if there is a third element that results in a consistent triplet.
In~\Cref{sec:given-econt}, we give a combinatorial description of all orientations $\varphi$ such that, for a given set~\Econt of containment edges, the pair $(\varphi, \Econt)$ is consistent.  
This yields an algorithm for computing a clique ordering $\sigma$ such that, for a given consistent pair $(\varphi, \Econt)$, the triplet $(\varphi, \sigma, \Econt)$ is consistent.
We prove the following.

\begin{restatable}{theorem}{thmrecognizeecont}
   \label{thm:recognize-econt}
    Given a mixed graph $G$ and $\Econt \subseteq E(G)$, there is a linear-time algorithm that decides whether $G$ admits an oriented interval representation where precisely the edges in~$\Econt$ are represented by containment.
\end{restatable}

A \emph{proper} interval graph is an interval graph that admits a representation where no interval properly contains another interval.
In other words, a proper interval graph is an interval graph with $\Econt=\emptyset$.  
A \emph{unit} interval graph is an interval graph that admits a representation where all intervals have the same length.
The classes of proper interval graphs and the class of unit interval graphs are the same~\cite{Roberts1969}.
Using the algorithm behind \cref{thm:recognize-econt}, we can recognize oriented proper (or unit) interval graphs.

\begin{corollary}
    Given a mixed graph $G$, we can decide in linear time whether $G$ admits an oriented proper (or unit) interval representation.
\end{corollary}

In~\Cref{sec:mpq-short}, we give a combinatorial description of all clique orderings $\sigma$ such that, for a given orientation $\varphi$, the pair $(\varphi, \sigma)$ is consistent. 
This allows us to design an algorithm for computing a set~\Econt of containment edges such that, for a given consistent pair $(\varphi, \sigma)$, the triplet~$(\varphi, \sigma, \Econt)$ is consistent.
Generalizing and speeding up the algorithm of Gutowski et al.~\cite{gmrswz-cmdig-GD22}, we prove the following.

\begin{restatable}{theorem}{thmgivenphigetrep}
   \label{thm:given-phi-get-rep}
   Given a mixed graph~$G$ with an orientation~$\varphi$, there is a linear-time algorithm that decides whether $G$ admits an oriented interval representation with orientation~$\varphi$.
\end{restatable}

\whenappendix{Statements whose proofs are in the appendix are marked with a clickable~($\star$).}


\section{Preliminaries}
\label{sec:prelims}


\subparagraph{Interval Graphs and Interval Representations.}
We consider sets of intervals on the real line in non-degenerate general position, i.e., all their endpoints are distinct.
Two intervals~$I$ and~$I'$ \emph{overlap} if their endpoints alternate, whereas $I$ \emph{contains} $I'$ if both endpoints of $I'$ lie between the two endpoints of~$I$.  
Given a set~\I of intervals, its \emph{intersection graph} is the graph that has \I as its vertex set and two different vertices $I$ and $I'$ are adjacent if and only if they have a non-empty intersection.
An undirected graph $G$ is called an \emph{interval graph} if~$G$ is isomorphic to an intersection graph of a set~\I of intervals. 
In this case, we call \I an \emph{interval representation} of~$G$, and, for each vertex~$u$ of~$G$, we let the interval in~\I that corresponds to $u$ be denoted by~$I_u$.
If~$I_u \cap I_v \ne \emptyset$, then we say that they \emph{realize} the corresponding edge of~$G$.  If~$I_u$ contains~$I_v$, then we say that the edge~$\{u,v\}$ is a \emph{containment edge}.

 
For $W\subseteq V(G)$, we let $G[W]$ denote the subgraph of~$G$ induced by the vertices of~$W$. 
Its induced representation is denoted by $\I[W]$.
For a graph~$G$, let $\mathcal C(G)$ be the set of maximal cliques of $G$.
It is well-known that a graph $G$ is an interval graph if and only if~$\mathcal C(G)$ admits a linear ordering in which, for every vertex~$u$ of~$G$, the set $\{C\in \mathcal C(G) \mid u \in C\}$ is consecutive~\cite{FulkersonGross1965}.
Such an ordering is called a \emph{clique order}; see \cref{fig:example-b}.

\begin{figure}
  \centering
  \begin{subcaptiongroup}
    \textsfbf{(a)} \includegraphics[page=5,valign=t]{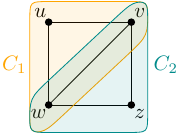}
    \phantomcaption\label{fig:example-a}\hfil
    \textsfbf{(b)} \includegraphics[page=16,valign=t]{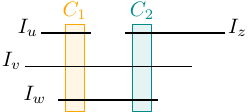} 
    \phantomcaption\label{fig:example-b}
  \end{subcaptiongroup}
  \caption{(a) A graph $G$ with $\mathcal C(G) = \{C_1, C_2\}$. (b)~An
    interval rep.~$\mathcal{I}$ of~$G$ and its clique order.}
  \label{fig:example}
\end{figure}

\subparagraph{Oriented Interval Representations.}
A graph containing directed arcs and undirected edges is a \emph{mixed graph}. Given a mixed graph~$G$, let $\Arcs(G)$ denote the set of arcs of~$G$, and let~$E(G)$ denote the set of edges of~$G$. For a mixed graph~$G$, we let~$U(G)$ denote the \emph{underlying undirected graph} that contains an edge~$\{u,v\}$ whenever~$G$ contains an edge or an arc with endpoints~$u$ and~$v$.  An \emph{orientation} of $G$ is a mapping~$\varphi \colon V(G) \to \{-1,1\}$, which, in the context of an interval representation~$\mathcal I = \{I_v \mid v \in V(G)\}$ of $G$, we interpret as an orientation of the respective intervals, where an interval $I_v \in \mathcal I$ is \emph{right-oriented} if~$\varphi(v) = +1$ and it is \emph{left-oriented} if~$\varphi(v) = -1$.
An \emph{interval representation} of a mixed graph $G$ is an interval representation of its underlying undirected graph~$U(G)$.  By contrast, an \emph{oriented} interval representation takes the directions of the arcs of $G$ into account as follows; see~\Cref{fig:orientation}.

\begin{definition}\label{def:oriented-interval-representation}
    Let~$G$ be a mixed graph.
    An \emph{oriented interval representation} of $G$ is a pair~$(\mathcal I, \varphi)$ consisting of an interval representation~$\mathcal I
    $ of $G$ and an orientation~$\varphi$ of~$G$ such that for each pair of intervals $I_u = [l_u,r_u]$ and $I_v = [l_v,r_v]$ from~\I with $l_u < l_v$:
    \begin{itemize}
        \item $(u,v) \in \Arcs(G)$ if and only if $I_u$ and $I_v$ overlap, and
        $\varphi(I_u)=\varphi(I_v)=+1$.
        \item $(v,u) \in \Arcs(G)$ if and only if $I_u$ and $I_v$ overlap, and
        $\varphi(I_u)=\varphi(I_v)=-1$.
        \item $\{u,v\} \in E(G)$ if and only if
        \begin{itemize}
            \item $I_u$ and~$I_v$ overlap and $\varphi(I_u) \ne \varphi(I_v)$ or
            \item $I_u$ contains $I_v$, regardless of the orientations of the intervals.
        \end{itemize}
    \end{itemize}
    We also call $\mathcal(I,\varphi)$~\emph{$\varphi$-oriented}.
\end{definition}

A mixed graph $G$ is an \emph{oriented interval graph}\footnote{Gutowski et al.~\cite{gmrswz-cmdig-GD22} called such graphs \emph{bidirectional} interval graphs.} if it admits an oriented interval representation. 
We call an interval representation~$\mathcal I$ of~$G$~\emph{$\varphi$-orientable} if $(\mathcal I,\varphi)$ is an oriented interval representation of $G$. 
\cref{fig:example-b} depicts an interval representation that is not $\varphi$-orientable for the orientation of the graph from \cref{fig:example1-a}.
A \emph{clique} of $G$ is a clique in its underlying undirected graph~$U(G)$, and we write~$\mathcal C(G)$ for~$\mathcal C(U(G))$.

\begin{toappendix}    
    \Cref{fig:orientation} illustrates the different types of edges
    and their interval representation.

    \begin{figure}[h]
        \centering
        \includegraphics[page=1]{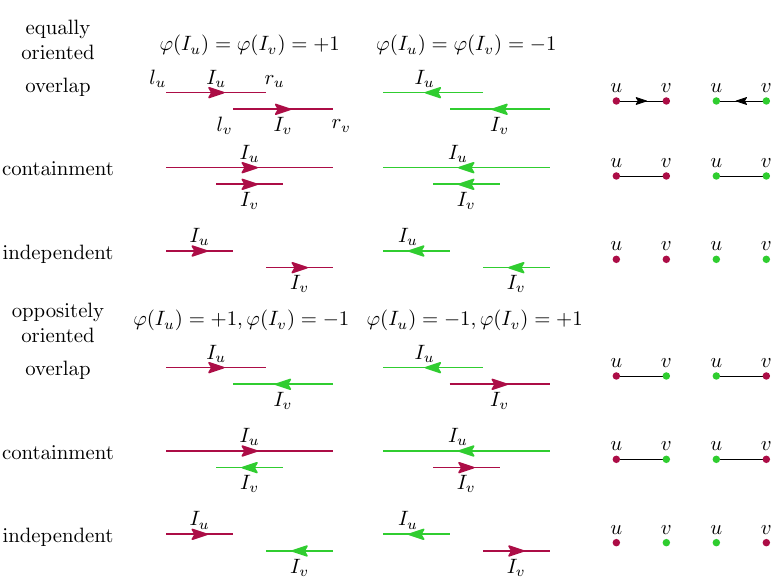}
        \caption{Possible positions of oriented intervals when $l_u < l_v$ and the edges they represent. In our illustrations, color ruby encodes right (positive) orientation, while lime means left (negative).} 
        \label{fig:orientation}
    \end{figure}
\end{toappendix}

\subparagraph{Matching Representations and Permutation Graphs.}
    A \emph{matching representation} of an undirected graph $G$ is a pair $(<_L,<_R)$ of linear orders of the vertex set $V(G)$ such
    that two vertices $u,v\in V(G)$ are ordered differently in $<_L$ and $<_R$ if and only if they are adjacent in~$G$.
    We depict such a representation as a diagram with the vertices in $<_L$ and $<_R$ placed on two parallel horizontal lines in their respective order and connecting the two representatives of a vertex $v$ on $<_L$ and $<_R$ by a segment; see \cref{fig:matching-interval-ex-c}.
    A graph that allows a matching representation is called a \emph{permutation graph}.  This is the case if and only if both~$G$ and its  complement $\overline G$ admit a \emph{transitive orientation}, i.e., an orientation of edges, where the presence of arcs $(u,v)$ and $(v,w)$ yields also the arc $(u,w)$. 
    In particular, such transitive orientations of $G$ and $\overline G$ can be obtained from the matching representation by the intersection (as relations) of $<_R$ with~$>_L$ and $<_R$ with $<_L$, respectively; in the diagram, the orientation of~$\overline G$ corresponds to the left-to-right ordering of non-intersecting segments~\cite{EvenPL1972}; see \cref{fig:matching-interval-ex-c}. 

    Moreover, by choosing $<_L$ and $<_R$ as the orders of left and right endpoints, respectively, every interval representation yields a matching representation of a permutation graph~$G$ whose edges represent nested intervals.
    A transitive orientation of $G$ therefore naturally orders the corresponding intervals by inclusion.




\section{Prescribed Containment Edges}
\label{sec:given-econt}

In this section we present an algorithm for testing whether a mixed graph $G$ with given containment edges $\Econt$ is an oriented interval graph.  This contains the case of oriented proper interval representations as a special case, where $\Econt = \emptyset$.

Our approach works in two steps.  First, we show that we can compute in linear time an orientation $\varphi$ of $G$ such that the pair~$(\varphi, \Econt)$ is consistent if and only if~$G$ admits an oriented interval representation with containment edges~$\Econt$.
We call such a $\varphi$ \emph{feasible}.
For a feasible~$\varphi$, we then check the existence of a representation.

\subsection{Computing a Feasible Orientation}
\label{subsec:computing-phi}

A key observation is that if $G-\Econt$ is connected, its feasible orientation~$\varphi$ (if any) is unique up to reversal as choosing the orientation of a vertex~$u$ determines the orientations of all its neighbors in~$G-\Econt$; namely, they have the same orientation if the two vertices are connected by an arc and they have opposite orientations if they are connected by an (undirected) edge.  Observe that such an orientation~$\varphi$ can be easily computed in linear time.  If~$G-\Econt$ is not connected, this reasoning can be applied to each individual connected component of~$G-\Econt$.  In the rest of this section, we show that these orientations, which are unique up to reversal, can be combined arbitrarily.

For an orientation $\varphi$ of $G$ and $W \subseteq V(G)$ define the \emph{reversal orientation $\rev_W(\varphi)$} by
$$ v \mapsto 
  \begin{cases}
     -\varphi(v) & v \in W \\
    \phantom{-}\varphi(v) & v \notin W
  \end{cases}
$$
as the orientation where precisely the vertices in $W$ have their orientation reversed.
Observe that for two disjoint sets $U, W \subseteq V(G)$ we have $\rev_U \circ \rev_W = \rev_{U \cup W}$, while for two sets in inclusion $U \subseteq W$ we have $\rev_U \circ \rev_W = \rev_{W \setminus U}$.

We are interested in conditions under which reversing the orientations of a subset~$W$ preserves the feasibility of orientations.  
Let $\mathcal I$ be a $\varphi$-orientable interval representation of~$G$.
We say that an interval $I$ (not necessarily in \I) is \emph{overlap-free with respect to \I} if there exists no interval $J \in \I$ that overlaps~$I$.
We define $\cont_\I (I) \coloneqq \{v \in V(G) \mid \I(v) \subseteq I\}$ as the set of vertices of $G$ whose intervals are contained in $I$.
We let~\emph{$\mirr_I(\mathcal I)$} denote the interval representation that we obtain from~$\mathcal I$
by mirroring the subrepresentation $\mathcal I[W]$ at the midpoint of $I$ where $W = \cont_\mathcal I(I)$.
We show that this mirroring of the subrepresentation~$\mathcal I[W]$ is consistent with reversing the orientations of the vertices in $W$, that is, $\mirr_I(\I)$ is a $\rev_{W}(\varphi)$-orientable interval representation of $G$.
  
\begin{restatable}{lemma}{lemrevrep}
   \restateref{lem:rev-rep}
   \label{lem:rev-rep}
    Let $G$ be a mixed graph and let \I be a $\varphi$-orientable interval representation of $G$.
    Let $I$ be an overlap-free interval with respect to \I and let $W = \cont_\I (I)$.
    Then $\mirr_I(\I)$ is a $\rev_{W}(\varphi)$-orientable interval representation of $G$.
\end{restatable}

\begin{toappendix}
\whenappendix{\lemrevrep*\label{lem:rev-rep*}}
\end{toappendix}

\begin{appendixproof}
    Treat for now $\mirr_I$ as a mapping assigning each interval $I_u\in \I$
    the value $\mirr_I(I_u)=(\mirr_I(\I))(u)$. 
    An interval $I_u$ is contained in $I$ if and only if $\mirr_I(I_u)$ is contained in $I$.
    For any two vertices~$u$ and~$v$ of~$G$, their intervals
    intersect in~\I if and only if they intersect in $\mirr_I(\I)$.
    Moreover, $I_u \subseteq I_v$ if and only if $\mirr_I(I_u) \subseteq \mirr_I(I_v)$.
    The only difference occurs with overlapping intervals.
    Namely, $I_u$ overlaps $I_v$ such that $I_u$ contains the left endpoint of $I_v$ if and only if $\mirr_I(I_u)$ overlaps $\mirr_I(I_v)$ such that $\mirr_I(I_u)$ contains the right endpoint of $\mirr_I(I_v)$.
    Hence, the order of overlaps is reversed in $\mirr_I(\I)[W]$ and therefore $\mirr_I(\I)$ is $\rev_{W}(\varphi)$-orientable.
\end{appendixproof}

We use this to show that reversing the orientations of any subset~$W$ of vertices such that all edges between $W$ and~$V \setminus W$ are containment edges, preserves the feasibility of an orientation.
We say that an interval~$I$ (not necessarily in \I) \emph{spans} the intervals of $\I[W]$ if every interval of $\I[W]$ is properly contained in $I$.

\begin{restatable}{lemma}{lemrevorient}
    \restateref{lem:rev-orient}
    \label{lem:rev-orient}
    Let $G$ be a mixed graph, and let $W \subseteq V$ be a vertex set such that all edges between the vertices of $W$ and $V \setminus W$ are containment edges.
    Then $G$ has a $\varphi$-orientable interval representation if and only if $G$ has a $\rev_{W}(\varphi)$-orientable interval representation.
\end{restatable}

\begin{toappendix}
\whenappendix{\lemrevorient*\label{lem:rev-orient*}}
\end{toappendix}

\begin{appendixproof}
    ``$\Rightarrow$'':
    Assume that $G$ admits a $\varphi$-orientable interval representation \I.
    Let $I$ be the smallest interval that spans $\I[W]$.
    Since the vertices of~$G-W$ are connected to $W$ only by containment edges, $I$ is overlap-free with respect to \I.
    Let $U = \cont_\I (I)$ and let $B_1, \dots, B_k$ be the vertex sets of the connected components of $G[U] - W$.
    For $i \in \{1, \dots k\}$, let $I_i$ be the smallest interval that spans $\I[B_i]$.
    Observe that $I_1, \dots, I_k$ are pairwise disjoint and overlap-free with respect to \I.
    Moreover, for every $i \in \{1, \dots, k\}$, we have that $\cont_\I (I_i) = B_i$.
    From \cref{lem:rev-rep} we know that $\mirr_I(\I)$ is a $\rev_{U}(\varphi)$-orientable interval representation of $G$.
    Let $I'_i$ be the interval obtained from mirroring $I_i$ at the midpoint of $I$.
    Then $I'_1, \dots, I'_k$ are again pairwise disjoint, and they are overlap-free with respect to $\mirr_I(\I)$.
    Moreover, we have $\cont_{\mirr_I(\I)}(I'_i) = B_i$.
    Since $I'_1, \dots, I'_k$ are pairwise disjoint, we can apply \cref{lem:rev-rep} to mirror them independently and obtain that $\mirr_{I'_1}(\mirr_{I'_2}(\dots(\mirr_{I'_k}(\mirr_I(\I)))\dots))$ is a $\rev_{B_1}(\rev_{B_2}(\dots (\rev_{B_k}(\rev_U(\varphi)))\dots))$-orientable interval representation of $G$.
    Since the $B_i$ are pairwise disjoint, we have that
    $\rev_{B_1}(\rev_{B_2}(\dots (\rev_{B_k}(\rev_U(\varphi)))\dots)) = \rev_{B}(\rev_U(\varphi))$, where $B = \bigcup_{i = 1}^k B_i$.
    Moreover, since $B \subseteq U$ and $W = U \setminus B$, we have that $\rev_{B}(\rev_U(\varphi)) = \rev_W(\varphi)$.

    ``$\Leftarrow$'': Observe that $\rev_W(\rev_W(\varphi)) = \varphi$, so it follows from the forward implication.
\end{appendixproof}

We show now that a feasible orientation (if one exists) can be computed in linear time.

\begin{lemma}
    \label{lem:linear-time-phi}
    There is a linear-time algorithm that given $G$ and $\Econt$ computes a feasible orientation $\varphi$ of $G$ if it admits one.
\end{lemma}
\begin{proof}
    Let $G$ be a mixed graph and let $K$ be a connected component of $G - \Econt$.
    Observe that since $K$ is connected by non-containment edges, choosing the orientation for a single vertex $v \in V(K)$ uniquely determines the orientations of all other vertices in $K$.
    That is, if it exists, the consistent orientation $\varphi_K$ of $K$ is unique up to reversal. We can compute $\varphi_K$ in time linear in the size of $K$.
    If $K$ has no consistent orientation, $G$ is a no-instance and we return an arbitrary orientation of $G$.
    Otherwise, we return the orientation~$\varphi$ of $G$ defined by $\varphi(v) = \varphi_K(v)$, where $K$ is the connected component of $G - \Econt$ that contains $v$.

    It remains to show that $\varphi$ is feasible.
    Clearly, if $G$ admits a $\varphi$-orientable interval representation, then $G$ is an oriented interval graph.
    Conversely, assume $G$ is an oriented interval graph and admits an  oriented interval representation $(\I,\varphi')$ for some orientation $\varphi'$.
    We choose~$(\I,\varphi')$ such that the number of vertices for which it coincides with $\varphi$ is maximized.
    If $\varphi = \varphi'$, then we are done.
    Otherwise, there is a vertex $v$ with $\varphi(v) \neq \varphi'(v)$.
    Let $K$ be the connected component of $G- \Econt$ that contains $v$.
    Since the orientation for each connected component~$K$ of~$G - \Econt$ is unique up to reversal, we have $\varphi(v) = -\varphi'(v)$ for all $v \in V(K)$.
    By \cref{lem:rev-orient} we have that since $G$ has a $\varphi'$-orientable interval representation, it has a $\rev_{K}(\varphi')$-orientable interval representation.
    Then $\rev_{K}(\varphi')$ coincides with $\varphi$ for more vertices than $\varphi'$.
    This is a contradiction to the choice of $\varphi'$.
\end{proof}

\subsection{Constructing an Interval Representation for a Given Orientation}
\label{subsec:constructing-phi}

    Based on \cref{lem:linear-time-phi} we may assume that, in addition to the mixed graph~$G$ and the containment edges~$\Econt$, we are also given a feasible orientation~$\varphi$.
    That is, the pair~$(\varphi, \Econt)$ is consistent if~$G$ admits an oriented interval representation with containment edges~\Econt.
    
    Our algorithm is based on the fact that the graph $H$ with $V(H) = V(G)$ and $E(H) = \Econt$ is a permutation graph that admits a special type of matching representation if and only if the given mixed interval graph $G$ admits an oriented interval representation.

    Let $u$ and $v$ be two vertices of~$G$ with $(u,v) \in \Arcs(G)$.
    Then any $\varphi$-orientable interval representation orders $u$ and $v$ in the sense that for $I_u = [l_u, r_u]$ and $I_v=[l_v, r_v]$, the order of~$l_u$ and $l_v$ is the same as the order of $r_u$ and $r_v$.
    If $\varphi(u) = \varphi(v) = 1$, then the direction of~$(u,v)$ implies that $l_u < l_v$ and $r_u < r_v$ in any $\varphi$-orientable interval representation.
    Similarly, $\varphi(u) = \varphi(v) = -1$ implies that $l_v < l_u$ and $r_v < r_u$ in any $\varphi$-orientable interval representation.
    We capture these conditions in the \emph{placement order $\prec_{\varphi}$} defined as
    \mbox{$u\prec_{\varphi} v$} if and only if $((u,v) \in \Arcs(G) \land \varphi(u) = \varphi(v) = 1)\lor((v,u) \in \Arcs(G) \land \varphi(u) = \varphi(v) = -1)$.
    By definition, $\prec_{\varphi}$ is antisymmetric and acyclic on $V(G)$. Note that $\prec_{\varphi}$ is not necessarily transitive, we have chosen this symbol and name as it well captures the relative position of two vertices connected by an arc.
    The following lemma establishes the connection between matching representations and containments in interval representations.


    \begin{restatable}{lemma}{lempermgraphto}
        \restateref{lem:perm-graph-to}
        \label{lem:perm-graph-to}
        Let $G$ be a mixed graph and let $(\I,\varphi)$ be its oriented interval representation with containment edges $\Econt$.
        Then $H =(V(G),\Econt)$ is a permutation graph and its
        complement $\overline{H}$ admits a transitive orientation $\prec_{\overline{H}}$ that extends $\prec_\varphi$.
    \end{restatable}

    \begin{toappendix}
    \whenappendix{\lempermgraphto*\label{lem:perm-graph-to*}}
    \end{toappendix}
    \begin{appendixproof}
        For $v \in V(G)$,
        let $I_v=[l_v, r_v]$ be the interval in \I that represents it.
        We first show that $H$ is a permutation graph.
        Let $<_L$ and $<_R$ be the two orderings of $V(G)$, in which the vertices are ordered by the left and right endpoints of the corresponding intervals in \I, respectively.
        We claim that $(<_L,<_R)$ is a matching representation of $H$.
        To see this, let $u$ and $v$ be two vertices of~$H$,
        and assume without loss of generality that $l_u < l_v$.
        
        If $\{u,v\} \in \Econt$, it follows that $I_v \subseteq I_u$, that is $l_u < l_v < r_v < r_u$.
        Therefore, we have $u <_L v$ and  $v <_R u$, and then $u$ intersects with $v$ in the matching diagram of $(<_L,<_R)$.
        On the other hand, if $\{u,v\} \notin \Econt$, then $r_u < r_v$, regardless of whether $\{u,v\} \in E \setminus \Econt$ or not.
        Therefore, we have $u <_L v$ and $u <_R v$, and then $u$ does not intersect with $v$ in the diagram;
        see \cref{fig:matching-interval-ex} for an example. Hence, $H$ is a permutation graph.
        
    \begin{figure}
        \begin{subcaptiongroup}
        \centering
        \textsfbf{(a)} \includegraphics[page=6,valign=t]{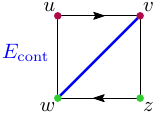}
        \phantomcaption \label{fig:matching-interval-ex-a}\hfil
        \textsfbf{(b)} \includegraphics[page=7,valign=t]{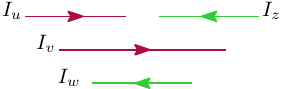}
        \phantomcaption\label{fig:matching-interval-ex-b}\hfil
        \textsfbf{(c)} \includegraphics[page=13,valign=t]{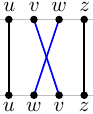}
        \phantomcaption\label{fig:matching-interval-ex-c}\hfil
        \textsfbf{(d)} \includegraphics[page=14,valign=t]{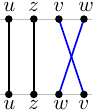}
        \phantomcaption\label{fig:matching-interval-ex-xc}
        \end{subcaptiongroup}
        
        \caption{(a) A mixed graph $G$ with set $\Econt = \big\{\{v,w\}\big\}$ of containment edges; (b)~an oriented interval representation of~$G$; (c)~diagram of a matching representation for $H = (V(G), \Econt)$ obtained from the interval representation; (d)~another matching representation for $H$ that corresponds to no interval representation.
        }
        \label{fig:matching-interval-ex}
        \end{figure}

        Let $\prec_{\overline{H}}$ be defined by $u\prec_{\overline{H}} v\Longleftrightarrow u <_L v 
        \land u <_R v$. It follows from the construction of $H$ and from the orders $<_L$ and $<_R$ imposed by \I that $\prec_{\overline{H}}$ is a transitive orientation of $\overline{H}$.

        We claim that $\prec_{\overline{H}}$ extends $\prec_\varphi$, i.e., the implication $u \prec_\varphi v \Rightarrow u \prec_{\overline{H}} v$ holds.
        For $u\prec_\varphi v$, we distinguish cases based on the orientations of $u$ and $v$.
        If~$\varphi(u) = \varphi(v) = 1$, then $(u,v) \in \Arcs(G)$.
        Therefore, $l_u < l_v < r_u < r_v$, which implies $u <_L v$ and $u <_R v$, and hence $u \prec_{\overline{H}} v$.
        Similarly, if $\varphi(u) = \varphi(v) = -1$, then $(v, u) \in \Arcs(G)$.
        Therefore, $l_v < l_u < r_v < r_u$, which implies $v <_L u$ and $v <_R u$, and hence,  in this case, we also get $u \prec_{\overline{H}} v$.
    \end{appendixproof}

    We note that Dushnik and Miller~\cite{Dushnik_1941} established the connection between containment interval graphs and posets of dimension 2.
    \cref{lem:perm-graph-to} gives necessary conditions for when $G$ is a $\varphi$-orientable interval graph, which are, however, not sufficient, as we may lose information about the non-containment adjacencies of $G$; see~\cref{ex:clique-counter}.
    \begin{toappendix}
    \begin{example}
    \label{ex:clique-counter}
    \cref{lem:perm-graph-to} gives necessary conditions for when $G$ is a $\varphi$-orientable interval graph, which are however not sufficient.
    
    Consider the mixed graphs $G_1$ and $G_2$ illustrated in \cref{fig:clique-counter-1-a}, with the only containment edge~$\{u,z\}$. By \cref{lem:linear-time-phi} all intervals can be assumed to be right-oriented in both graphs.
    
    \begin{figure}
        \centering
        \begin{subcaptiongroup}
        \textsfbf{(a)} \ \includegraphics[page=2,valign=t]{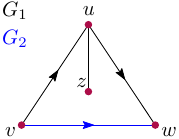}
        \phantomcaption\label{fig:clique-counter-1-a}
        \textsfbf{(b)}\includegraphics[page=7,valign=t]{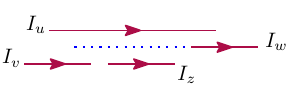} 
        \phantomcaption\label{fig:clique-counter-1-b}\hfil
        \textsfbf{(c)} \ \includegraphics[page=8,valign=t]{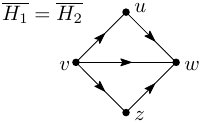} 
        \phantomcaption\label{fig:clique-counter-1-c}
        \end{subcaptiongroup}
 
        \caption{(a) A mixed graph $G_1$ with all intervals right-oriented. The mixed graph $G_2$ additionally contains the blue arc $(v,w)$.
        (b) A~$\varphi$-oriented interval representation of $G_1$. For $G_2$, there exists no way to realize the containment edge $\{u,z\}$ when $I_v$ and $I_w$ overlap (blue dots).
        (c) The graphs $G_1$ and $G_2$ yield the graphs $H_1$ and $H_2$, respectively, with their isomorphic complement graphs $\overline{H_1}$ and $\overline{H_2}$ both satisfying the properties of \cref{lem:perm-graph-to}, even though $G_2$ is a no-instance.}
        \label{fig:clique-counter-1}
    \end{figure}
    The representation depicted in \cref{fig:clique-counter-1-b} certifies that $G_1$ is an oriented interval graph. On the other hand, $G_2$ is not in this class:
    The vertices $\{u,v,w\}$ form a clique, while only $u$ is adjacent to $z$ via a containment edge.
    The edges $(v,u)$ and $(u,w)$ imply that $I_v$ and $I_w$ each contain one endpoint of $I_u$, respectively.
    As $I_v \cap I_w \neq \emptyset$, there is no point in $I_u$ that is not also contained in either $I_v$ or $I_w$.  Thus if~$I_z$ is contained in $I_u$, it will necessarily intersect $I_v$ or $I_w$, which is not allowed; 
    see \cref{fig:clique-counter-1-b}.

    Even though only one of $G_1$ and $G_2$ is an oriented interval graph with these containments, they both yield the same permutation graph $H$ that satisfies the properties of \cref{lem:perm-graph-to}; see \cref{fig:clique-counter-1-c}.
    In $H$, we lose information about the non-containment adjacencies of $G$.
    \end{example}
    \end{toappendix}
    However, this information is crucial for deciding if $G$ admits an oriented interval representation.
    In order to encode this information in $H$ we introduce the following construction.

    For a mixed graph $G$, the \emph{augmented graph $G^+$} is obtained by adding a clique point vertex~$c_i$ for each $C_i \in \mathcal C(G)$ to $G$,
    such that $c_i$ is adjacent to a vertex $v \in G$ if and only if~$v \in C_i$.
    We consider the edges incident to clique point vertices as elements of the containment edges~$\ECcont$ of $G^+$.
    The orientation of the clique point vertices is thus irrelevant and we extend $\varphi$ to an orientation $\varphi^+$ of $G^+$ by setting $\varphi^+(c_i) \coloneqq 1$ for every $c_i$.

    The \emph{augmented interval representation $\I^+$} is obtained from \I by adding $I_{c_i}$ strictly inside~$\bigcap \I[C_i]$ for every $c_i$.
    Sets $C_1^+ = C_1 \cup \{c_1\}, \dots, C_k^+ = C_k \cup \{c_k\}$ are the \emph{augmented maximal cliques}; see \cref{fig:example2}.
    In the augmented representation, the augmented maximal cliques are ordered the same way as the original ones and with a slight abuse of notation, we also denote this ordering by $\sigma$.

    \begin{figure}
        \begin{subcaptiongroup}
        \centering
        \textsfbf{(a)} \includegraphics[page=3,valign=t]{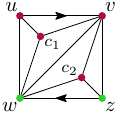}
        \phantomcaption\label{fig:example2-a}
        \textsfbf{(b)} \includegraphics[page=8,valign=t]{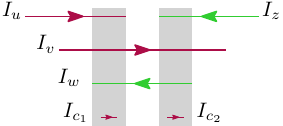} 
        \phantomcaption\label{fig:example2-b}
        \textsfbf{(c)} \includegraphics[page=15,valign=t]{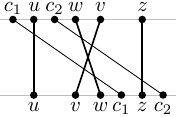}
        \phantomcaption\label{fig:example2-c}
        \textsfbf{(d)} \includegraphics[page=12,valign=t]{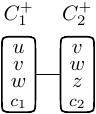}
        \phantomcaption\label{fig:example2-d}
        \end{subcaptiongroup}
        
        \caption{(a) The augmentation~$G^+$ of the graph~$G$ from \cref{fig:matching-interval-ex-a}, (b)~an interval representation $\mathcal{I}^+$ of $G^+$, (c) diagram of a matching representation of $H^+$, and (d)~the clique order of $\mathcal{I}^+$.
        }
        \label{fig:example2}
    \end{figure}

    \begin{restatable}{lemma}{lemggc}
    \restateref{lem:g-gc}
        \label{lem:g-gc}
        A graph $G$ has a $\varphi$-oriented interval representation if and only if $G^+$ has a $\varphi^+$-oriented
        interval representation such that no clique point interval properly contains another interval.
    \end{restatable}

    \begin{toappendix}
    \whenappendix{\lemggc*\label{lem:g-gc*}}
    \end{toappendix}

    \begin{appendixproof}
        A $\varphi$-oriented interval representation of $G$ can be obtained from a $\varphi^+$-oriented interval representation of $G^+$ by removing the intervals that correspond to clique point vertices.
        Conversely, let \I be a $\varphi$-oriented interval representation of $G$.
        Then for each maximal clique~$C_i$ of~$G$ there is a point $p_i$ that is contained in exactly the intervals of~$C_i$.
        We obtain a $\varphi^+$-oriented interval representation $(\I^+,\varphi^+)$ of $G^+$ by inserting for each $c_i$ a corresponding arbitrarily small interval that contains $p_i$ and is contained in all intervals of $C_i$.
        Observe that such an interval does not contain any other interval of the representation.
    \end{appendixproof}

    Assume $G$ is an oriented interval graph with an oriented interval representation $(\I,\varphi)$.
    By \cref{lem:g-gc} we obtain that $G^+$ admits an oriented interval representation $(\I^+,\varphi^+)$,  where every clique point interval is inclusion-wise minimal.
    By a slight abuse of notation we use the symbol~$\prec_\varphi$ also for the placement order on $G^+$, because clique point vertices do not participate in any of its comparisons,
    and hence both placement orders on $G$ and $G^+$ are identical as~sets. 

    Let $\augH$ be the \emph{containment graph} of $G^+$,
    i.e., $V(\augH) = V(G^+)$ and $E(\augH) = \ECcont$.
    By~\cref{lem:perm-graph-to}, $\augH$ is a permutation graph whose complement $\overline{\augH}$ has a transitive orientation~$\precoverlineaugH$ that extends~$\prec_{\varphi}$.
    By orienting $\augH$ by inclusion of the corresponding intervals, i.e., choosing $ u\precaugH v \Longleftrightarrow I_u \subset I_v$, we
    obtain a transitive orientation $\precaugH$ of $\augH$ such that $c_i \precaugH v$ for all $v \in C_i$.
   We show that for $\augH$, these necessary conditions are also sufficient.

    \begin{restatable}{theorem}{thmorientable}
        \restateref{thm:orientable}
        \label{thm:orientable}
        Let $\varphi$ be an orientation of a mixed interval graph~$G$.
        Let $\augH = (V^+, \ECcont)$ be the containment graph of $G^+$.
        Then $G$ admits a $\varphi$-oriented interval representation if and only if
        \begin{enumerate}[(i)]
            \item $\augH$ admits a transitive orientation $\precaugH$ such that each $c_i$ is a minimal element and
            \item $\overline{\augH}$ admits a transitive orientation $\precoverlineaugH$ that extends $\prec_\varphi$.
        \end{enumerate}
        
    \end{restatable}

    \begin{toappendix}
    \whenappendix{\thmorientable*\label{thm:orientable*}}
    \end{toappendix}

    \begin{appendixproof}
        The necessity of conditions (i) and (ii) follows from the discussion above.
        For the sufficiency, assume conditions (i) and (ii) hold. 
        From a routine case analysis on comparable pairs it follows that the union of  relations~$\precaugH$ and~$\precoverlineaugH$ is transitive, as well as the union of~$\succaugH$ and~$\precoverlineaugH$~\cite{golumbic2004algorithmic}. As any distinct $u,v\in V^+$ are comparable either by~$\precaugH$ or by~$\precoverlineaugH$ we obtain two linear orders $<_R$ and $<_L$ on $V^+$ defined as $u<_R v \Longleftrightarrow u \precaugH v \lor u \precoverlineaugH v$ and $u<_L v \Longleftrightarrow u \succaugH v \lor u \precoverlineaugH v$.
        We show that $<_R$ and $<_L$ can be used to construct a $\varphi^+$-oriented interval representation of $G^+$ such that no $I_{c_i}$ properly contains another interval of the representation.

        Condition (i) implies $v <_R c_i$ and $c_i <_L v$ for each $v \in C_i$.
        Since the clique point vertices are not adjacent in $\augH$, they are linearly ordered by $\precoverlineaugH$ and thus their order is the same in both $<_L$ and $<_R$.
        In the following we assume that the clique point vertices~$\{c_1, \dots, c_k\}$ are numbered in this order.
        This yields the following structure.
        Let $L_1$ be the set of all vertices with $v <_L c_1$ and $L_i = \{v \mid c_{i-1} <_L v <_L c_i\}$ for $i\in \{2,\dots,k\}$.
        Similarly, let $R_k$ be the set of all vertices with $c_k <_R v$ and $R_j = \{v \mid c_j <_R v <_R c_{j+1}\}$ for $j\in\{ 1,\dots, k-1\}$.
        The sets $L_1, \dots, L_k$ and $R_1, \dots, R_k$ are ordered according to $<_L$ and $<_R$, respectively; see \cref{fig:sections}.

        \begin{figure}
            \centering
            \includegraphics[]{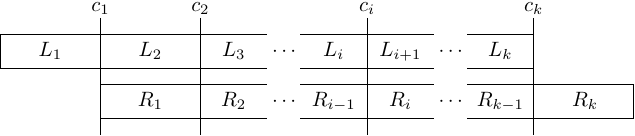}
            \caption{Partitioning $V(G)$ into $L_1,\dots,L_k$ and 
            into $R_1,\dots,R_k$ using the clique point vertices.}
            \label{fig:sections}
        \end{figure}

        Consider the sequence $L_1 c_1 c_1 R_1 L_2 c_2 c_2 R_2 \dots L_k c_k c_k R_k$, which contains each vertex of~$V^+$ exactly twice.
        We construct an interval representation $\I^+ = \{I_v\mid v \in V^+\}$ by choosing for each vertex $v$ the interval $I_v = [l_v, r_v]$ where $l_v$ and $r_v$ denote the position of the first and the second occurrence of $v$ in this sequence.
        To argue that $\I^+$ is a $\varphi^+$-orientable interval representation of $G^+$,
        we verify that $\I^+$  correctly represents all its edges, arcs, and non-edges.

        Consider $\{u,v\} \in \ECcont$ and without loss of generality, assume $u\precaugH v$.
        Then we have $v <_L u$ and $u <_R v$.
        This implies $l_v < l_u < r_u < r_v$, which is the desired containment~$I_u \subseteq I_v$.
    
        Moreover, a vertex $v\in V$ is adjacent to a clique point vertex $c_i$ if and only if $v <_L c_i <_R v$, so if $v$ and $c_i$ are non-adjacent then either $l_v > l_{c_i}$ or $r_{c_i}> r_v$. Hence the intervals~$I_v$ and~$I_{c_i}$ are disjoint, because no other interval have endpoints inside $I_{c_i}$.
          
        We showed that adjacencies related to clique-point vertices are represented properly, so we restrict our attention to representation of the graph $G$ only. Recall that the edges of \Econt were already discussed.

        Assume that $u$ and $v$ are not adjacent in $G$.
        Let $c_i, c_j$ be two clique point vertices that are adjacent to $u$ and~$v$, respectively, such that $|i - j|$ is minimized.
        Since $u$ and $v$ are not adjacent, we have~$\{u,c_j\} \notin \ECcont$, and $\{v,c_i\} \notin \ECcont$.
        In particular, $c_i \neq c_j$, so we without loss of generality assume that $i < j$, and in immediate consequence $c_i <_L c_j$ and $c_i <_R c_j$.
        By minimality we have $r_u < l_{c_{i+1}} $ and $r_{c_{j-1}} < l_v$. Even if $i$ and $j$ were consecutive, as $R_i$ is placed before $L_{i+1}$ in the sequence, we get that $r_u < l_v$. Hence $I_u \cap I_v = \emptyset$ as required.

        Assume that $u$ and $v$ form an arc or a non-containment edge.
        Since $\{u,v\} \notin E(\augH)$, similar to the clique point vertices above, the order of $u$ and $v$ is the same in both~$<_L$ and~$<_R$.
        We assume without loss of generality $u \precoverlineaugH v$ and hence $u <_L v$ and $u <_R v$.
        Since~$u$ and~$v$ are adjacent in $G$, there exists a clique point vertex $c_i$ with $u,v \in C_i$ and $u <_L v <_L c_i$ and $c_i <_R u <_R v$.
        Therefore, $l_u < l_v < l_{c_i} < r_{c_i} < r_u < r_v$, which implies that $I_u$ and $I_v$ overlap and $I_u$ starts first.

        We now argue that this overlap together with $\varphi$ gives a correct representation of
        an edge between $u$ and $v$.
        If $\{u,v\}$ is an undirected edge of $G$, then by assumption $\varphi(u) \neq \varphi(v)$ and 
        $\{u,v\}$ is correctly represented by the intervals $I_u$ and $I_v$ with orientation~$\varphi(u) \neq \varphi(v)$.
        If $(u,v)$ is an arc in $G$, then by assumption $\varphi(u) = \varphi(v)$.
        As by condition (ii) the order~$\precoverlineaugH$ extends~$\prec_\varphi$, the assumption $u \precoverlineaugH v$ and the
        definition of $\prec_\varphi$ imply $\varphi(u) = \varphi(v) = 1$, therefore the arc is represented correctly.
        Similarly, if $(v,u)\in \Arcs(G)$, then $\varphi(u) = \varphi(v) = -1$, and again the arc representation is correct.
    \end{appendixproof}

    The following theorem is the main result of this section.

\thmrecognizeecont*

\begin{proof}
    By \cref{lem:linear-time-phi} we compute a feasible orientation $\varphi$ of $G$ in linear time.
    Thus, $G$ is an oriented interval graph if and only if it has a $\varphi$-oriented interval representation.
    Recall that by \cref{thm:orientable}, $G$ is a $\varphi$-orientable interval graph, if and only if (i)
        $\augH=(V^+,\ECcont)$ admits a transitive orientation $\precaugH$ such that each $c_i$ is a minimal element and (ii) $\overline{\augH}$ admits a transitive orientation $\precoverlineaugH$ that extends $\prec_\varphi$.
    The problem of extending a partial transitive orientation of a graph with $n$ vertices and $m$ edges can be solved in $O(n+m)$ time~\cite{Muench2024}.
    In our setting, the difficulty with linear time lies in checking the existence of $\precoverlineaugH$ as the complement graph $\overline{\augH}$ may be too dense.
    To solve this issue, we reduce checking the existence of $\precaugH$ and $\precoverlineaugH$ to checking the existence 
    of a matching representation of $\augH$ that satisfies certain constraints.
    We then solve the latter problem by using lowest common ancestors and modular decompositions; see \cref{subsec:linear-time-recog} for the details.
\end{proof}

We note that the same characterization and algorithms apply to the case where, in addition, we prescribe a clique ordering~$\sigma = C_1 < C_2 < \cdots < C_k$ of the required interval representation, as it suffices to add~$(c_1,c_2), (c_2,c_3),\dots, (c_{k-1},c_k)$ to~$\prec_\varphi$ in \cref{thm:orientable}.  Since the additional constraints are of the same type, the linear-time algorithm from \cref{subsec:linear-time-recog} can handle them without modification.


\begin{corollary}
  \label{cor:linear-time-given-sigma}
  Let $G$ be a mixed graph with containment edges
  $\Econt \subseteq E$, let~$\varphi$ be an orientation of~$G$, and
  let $\sigma$ be a clique order of~$G$.  There is a linear-time
  algorithm that decides whether $G$ admits a $\varphi$-oriented
  interval representation with clique order~$\sigma$ where precisely the edges in~$\Econt$ are represented by interval inclusions.
\end{corollary}

\begin{toappendix}
\subsection{Linear-Time Algorithm}
\label{subsec:linear-time-recog}

As our algorithm builds on the so-called modular decomposition, we first give a description of this structure.

\subparagraph{Modular Decomposition.}
Let $H$ be an undirected graph.
A \emph{module} of~$H$ is a non-empty set~$M \subseteq V(H)$ of vertices such that every vertex $u \in V(H) \setminus M$ is either adjacent to every vertex in~$M$ or to none of them.
Note that $V(H)$ itself as well as the singleton subsets are modules; we call them the~\emph{trivial} modules of~$H$. 
A module~$M \subsetneq V$ is called~\emph{maximal} if no nontrivial module contains $M$ as a proper subset.

Let~$T$ be a rooted tree with leaf set~$V(H)$ such that, for every node~$\mu$ in~$T$, the leaf set~$M_\mu$ of the subtree of~$T$ rooted at~$\mu$ is a module of~$H$.
We associate with each inner node~$\mu$ in~$T$ a \emph{quotient graph}~$H_\mu$ which is obtained from~the subgraph $H[M_\mu]$ of~$H$ induced by $M_\mu$
by contracting $M_\nu$ for each child~$\nu$ of~$\mu$ into a single vertex; see \cref{fig:modular-decomoposition}.
We call a graph \emph{prime} if it has at least three vertices and no non-trivial modules.
We call a node~$\mu$ in~$T$ \emph{empty}, \emph{complete},
or~\emph{prime} if its corresponding quotient graph is edgeless,
complete, or prime, respectively.
The~\emph{modular decomposition} of~$H$, introduced by
Gallai~\cite{gallai1967transitiv}, is the uniquely defined rooted tree~$T$ with leaf set~$V(H)$ such that (i)~for every node~$\mu$ in~$T$, $M_\mu$ is a module of~$H$, (ii)~every quotient graph is empty, complete or prime and (iii)~no two adjacent nodes are both empty or both complete; see \cref{fig:modular-decomoposition}.
For every non-prime node~$\mu$ with children $\nu_1, \dots, \nu_k$, the sets $M_{\nu_1}, \dots, M_{\nu_k}$ are exactly the maximal modules in $H[M_\mu]$.
For a node~$\mu$ in $T$ and a vertex $v \in M_\mu$ let $\rep_{\mu}(v)$ denote vertex of $H_\mu$ which corresponds to the child~$\nu$ of~$\mu$ such that $v \in M_\nu$.
The modular decomposition of a graph $H$ can be computed
in~$O(|V(H)|+|E(H)|)$ time~\cite{McConnell1999,corneil2024recursivelineartimemodular}.

\begin{figure}
    \centering
        \textsfbf{(a)} \ \includegraphics[page=1,valign=t]{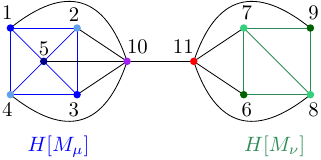} \hfil
        \textsfbf{(b)} \includegraphics[page=2,valign=t]{fig/perm-graph-decomp-2.pdf}
    \caption{(a) A graph~$H$ with two modules emphasized $M_\mu=\{1,\dots,5\}$ and $M_\nu=\{6,\dots,9\}$; (b) its modular decomposition together with the quotient graphs. Only the root is a prime node, while three nodes are complete, e.g. $\mu$ and $\nu$ and three are empty.}
    \label{fig:modular-decomoposition}
\end{figure}

\subparagraph{The Algorithm.}

We use the fact that the conditions from \cref{thm:orientable} can be rephrased as a constrained representation problem on permutation graphs to devise a linear-time algorithm.
Let~$H$ be a permutation graph.
    A \emph{constraint} is a pair $(u,v)$ of two vertices of $H$ and a representation $D = (<_L, <_R)$ of $H$
    \emph{satisfies} $(u,v)$ if $u <_R v$.
    A representation satisfies a set of constraints if it satisfies every constraint in the set.

    To check the existence of a matching representation that satisfies a set of constraints~$\Gamma$, we use the modular decomposition $T$ of $H$.
    Münch et al.~\cite[Theorem 5]{Muench2024} show that the choice of a matching representation of $H$ is equivalent to the simultaneous choice of a representation for the quotient graph $H_\mu$ for each node $\mu$ in $T$.
    For each node~$\mu$ in $T$ let~$D_\mu$ denote a matching representation of the quotient graph~$H_\mu$. 
    Graphically, to obtain a matching representation for $H[M_\mu]$, consider 
    the diagram of $D_\mu$ and for each child $\nu$ of $\mu$, replace the line segment that represents $\nu$ by the diagram of the matching representation of $H[M_\nu]$; see \cref{fig:perm-graph-decomp}.  Conversely, a matching representation $D$ of $H$ yields a matching representation~$D_\mu$ of every~$H_\mu$.

    \begin{figure}
        \centering
        \textsfbf{(a)} \includegraphics[page=3,valign=t]{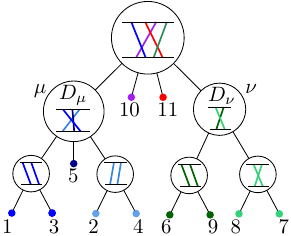} \hfil  
        \textsfbf{(b)} \ \includegraphics[page=4,valign=t]{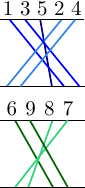} \hfil
        \textsfbf{(c)} \ \includegraphics[page=5,valign=t]{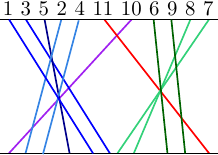} 
        \caption{$(a)$ Assignment of matching representations to the quotient graphs in $T$. $(b)$ The corresponding matching representations for $H[M_\mu]$ and $H[M_\nu]$. $(c)$ The matching representation of $H$ composed from the matching representations of the quotient graphs.}
        \label{fig:perm-graph-decomp}
    \end{figure}

    We show that the constraints in $\Gamma$ can be expressed as constraints $\Gamma_\mu$ on the representations of the quotient graphs $H_\mu$.
    Since the representations of quotient graphs have a simple structure, we can check efficiently whether each of them admits a representation that satisfies the constraints $\Gamma_\mu$.
    If this is the case, combining them as described above yields the desired matching representation of $H$.
    We now describe this procedure in detail.

    Let $D = (<_L, <_R)$ be a matching representation of $H$ and for each node~$\mu$ in~$T$ let $D_\mu = (<_L^\mu, <_R^\mu)$ be the corresponding matching representation of $H_\mu$. 
    The key observation is that for any two vertices $u,v$ of $H$ the following equivalence holds: $u <_R v$ if and only if
    $\rep_\mu(u) <_R^\mu \rep_\mu(v)$, where~$\mu$ corresponds to the smallest module $M_\mu$
    containing both $u$ and $v$. This $\mu$ is the lowest common ancestor $\mathrm{lca}(u,v)$ of the leaves $u$ and $v$ in the tree $T$.
    Thus we can partition the constraints in~$\Gamma$ according to the quotient graph they restrict.
    Let $\Gamma_\mu \coloneqq \{(\rep_\mu(u), \rep_\mu(v)) \mid (u,v) \in \Gamma \land \mu = \mathrm{lca}(u,v)\}$. 
    From the previous discussion, we immediately get the following statement.

    \begin{lemma}
        \label{lem:sat-mu-global-local}
        $D$ satisfies $\Gamma$ if and only if each $D_\mu$ satisfies~$\Gamma_\mu$.
    \end{lemma}
 
    We now show how to test in linear time if a permutation graph $H$ admits a matching representation that satisfies $\Gamma$.

    \begin{lemma}
        \label{lem:match-constr}
        Let $H=(V,E)$ be a permutation graph and let $\Gamma$ be a set of constraints.
        There is an algorithm that decides in time $O(|V|+|E|+|\Gamma|)$, whether $H$ admits a matching representation that satisfies~$\Gamma$.
    \end{lemma}
    \begin{proof}
        We first compute the modular decomposition $T$ of $H$ in linear time~\cite{McConnell1999}.
        By \cref{lem:sat-mu-global-local} it suffices to test for each node $\mu$ of $T$ if $H_\mu$ admits a representation $D_\mu$ that satisfies $\Gamma_\mu$.
        We argue that such a test can be performed in time linear in the size of $\Gamma_\mu$ and~$H_\mu$.

        We distinguish the three types of nodes in $T$: complete, empty and prime.
        If $\mu$ is complete or empty, then any order $<_R^\mu$ of the vertices of $H_\mu$ uniquely defines a matching
        representation~$(<_L^\mu, <_R^\mu)$ of $H_\mu$: If $\mu$ is empty, then $<_L^\mu$ must be identical with $<_R^\mu$, while if~$\mu$ is complete then $<_L^\mu$ must be the
        reverse of $<_R^\mu$.
        It thus suffices to topologically sort the vertices of $H_\mu$ according to $\Gamma_\mu$, which can be done in the claimed time.
        If $\mu$ is prime, then it is known that $H_\mu$ admits only four distinct matching representations~\cite{Muench2024} and for each of them, it is possible to test whether it satisfies $\Gamma_\mu$ in the claimed running time.

        By~\cite[Lemma 1]{Muench2024}, all sets $\Gamma_\mu$ can be computed in time $O(|\Gamma|)$. Moreover, $\sum_{\mu \in V(T)} |\Gamma_\mu| = |\Gamma|$ and $\sum_{\mu \in V(T)} |V(H_\mu)|+|E(H_\mu)| = O(|V|+|E|)$, hence  the time to process all nodes of $T$ is linear in $|V| + |E| + |\Gamma|$.
    \end{proof}

    We now apply this in the context of oriented interval graphs.


    \thmrecognizeecont*       
   
    \begin{proof}
        Using \cref{lem:linear-time-phi}, we can compute a feasible orientation~$\varphi$, if~$G$ admits one. It then suffices to check whether~$G$ is a~$\varphi$-orientable interval graph.  Recall that by \cref{thm:orientable}, this is the case, if and only if 
        (i) $\augH=(V^+,\ECcont)$ admits a transitive orientation $\precaugH$ such that each $c_i$ is a minimal element and 
        (ii) $\overline{\augH}$ admits a transitive orientation $\precoverlineaugH$ that extends $\prec_\varphi$. We claim that these conditions are equivalent to the existence of a matching representation $D=(<_L,<_R)$ of $\augH$ that satisfies the constraints $\Gamma = \{(u,v)\mid u \prec_\varphi v \} \cup \{(c_i,v) \mid v\in C_i\in\mathcal{C}(G^+) \}$.

        Recall that the two linear orders $<_R$ and $<_L$ are determined by formulae  $u<_R v \Longleftrightarrow u \precaugH v \lor u \precoverlineaugH v$ and $u<_L v \Longleftrightarrow u \succaugH v \lor u \precoverlineaugH v$. In particular, if $\precoverlineaugH$ extends $\prec_\varphi$, then also $<_R$ extends $\prec_\varphi$, and hence $\{(u,v)\mid u \prec_\varphi v \} \subseteq \Gamma$. Analogously, if each $c_i$ is a minimal element in $\precaugH$, then we get $c_i<_R v$ for each clique point $c_i$ and $v\in C_i$.%

        Conversely, if constraints $\{(c_i,v) \mid v\in C_i\in\mathcal{C}(G^+)\}$ are satisfied in $D$, then $c_i<_R v$ for each $v\in C_i$, hence $c_i$ is a source in the order $<_R$, as $c_i$ does not participate in any other edges of $\precaugH$.
        Analogously, if the constraints $\{(u,v)\mid u \prec_\varphi v \}$ are satisfied, the order $<_R$ extends $\prec_\varphi$. Moreover the transitive orientation of $\overline{H}$

        is the intersection of the linear orders $<_L$ and $<_R$. In addition, when $u$ and $v$ are comparable in 
        $\precoverlineaugH$, they are non-adjacent in $\augH$, so $u<_R v$, hence also $\precoverlineaugH$ extends $\prec_\varphi$ as required by the condition (ii).

        Regarding the time complexity, \cref{lem:match-constr} guarantees that the existence of $D$ can be tested in time $O(|V^+| + |\ECcont| + |\Gamma|)$.
        Recall that $\augH$ is a subgraph of $G^+$, which is obtained from $G$ by adding clique point vertices. As each maximal clique $C_i$ is the first in $\sigma$ for some vertex as well as the last for another one, we have $|V^+| \le 2|V|$. Similarly, $|\ECcont|\le |V|+|E|+|\Arcs|$, because the number of edges added by inserting $c_i$ to $C_i$ is by one greater than the number of edges incident in $C_i$ with any vertex that has there its first occurrence.
        Moreover, since every pair comparable by $\prec_\varphi$ corresponds bijectively to an arc of \Arcs, it follows that $|\Gamma| \in O(|E|+|\Arcs|)$ as well.    
    \end{proof}
\end{toappendix}

\section{Prescribed Orientations}
\label{sec:mpq-short}

    In this section, we assume that $\varphi$ is given as part of the input.      The full version of this section can be found in~\Cref{sec:given-phi}.
    We outline an algorithm that recognizes $\varphi$-orientable interval graphs using clique orderings.
    The proof works in two steps.
    First, we show that we can compute a clique
    ordering~$\sigma$ of~$G$ such that $G$ has a $\varphi$-oriented
    interval representation~\I if and only if the pair~$(\varphi, \sigma)$ is consistent.
    Then, we construct a $\varphi$-orientable interval representation of
    $G$ with clique ordering $\sigma$ if and only if $G$ is a $\varphi$-oriented interval graph.
    
    A \emph{PQ-tree}~$T$ over a finite universe~$U$ is a rooted ordered tree whose leaves correspond to the elements of~$U$ and whose internal nodes are of two types called \emph{P-} and \emph{Q-nodes}. Such a PQ-tree represents all permutations of its leaves that can be obtained by arbitrarily reordering the children of P-nodes and reversing the order of children of any set of Q-nodes.  The PQ-tree of an interval graph $G$ is the PQ-tree that has as leaves the maximal cliques of $G$ and represents the clique orders of $G$~\cite{BoothL1976}.
    We call the edges of the PQ-tree \emph{links} and for a Q-node $x$ with children~$y_1,\dots,y_k$ (in this order), we call a set of consecutive links $S_{i,j} = \{xy_i, \dots, xy_j\}$ a \emph{segment} of~$x$. 

     \begin{figure}[b]
        \begin{subcaptiongroup}
        \begin{minipage}{.4\textwidth}          
        \textsfbf{(a)}\includegraphics[page=1,valign=t]{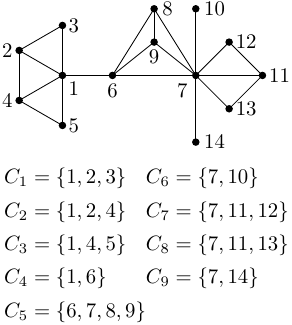}
        \phantomcaption\label{fig:MPQ-tree-example-a}\hfil
        \end{minipage}%
        \begin{minipage}{.55\textwidth}
        \textsfbf{(b)}\hspace{2.5mm}\includegraphics[page=3,valign=t]{fig/MPQ-tree-example} 
        \phantomcaption\label{fig:MPQ-tree-example-b}\hfil       
        \medskip
        
        \textsfbf{(c)} \includegraphics[page=4,valign=t]{fig/MPQ-tree-example}
        \phantomcaption\label{fig:MPQ-tree-example-c}\hfil
        \end{minipage}
        \end{subcaptiongroup}
        \caption{
        (a) An interval graph $G$ and its maximal cliques; (b)~an MPQ-tree $T$ of $G$ (P-nodes are round, Q-nodes are boxes); (c)~the interval representation corresp.\ to the ordering of the leaves in $T$.
        }
        \label{fig:MPQ-tree-example}
    \end{figure}

    To describe the clique orderings of $G$ we use modified PQ-trees, also called MPQ-trees~\cite{Korte1989}; see~\Cref{fig:MPQ-tree-example} for an example.
    An \emph{MPQ-tree} of an interval graph $G$ consists of the PQ-tree of $G$ with additional annotations that associate the vertices of $G$ (and not just its maximal cliques) with certain parts of the PQ-tree.   Namely,
    we assign every vertex $v$ in $G$ to the node $x$ in $T$ that is the lowest common ancestor of all cliques containing~$v$.  For a node~$x$ of the MPQ-tree, we let~$W_x$ denote the set of vertices assigned to~$x$. In~\Cref{fig:MPQ-tree-example-b}, for example, we have $W_{x_1} = \{1, 6, 7\}$ and $W_{x_4} = \{11\}$.

    Let $x$ be a Q-node~with children~$y_1,\dots,y_k$ and let~$v$ be a vertex that is assigned to~$x$.  It follows from the construction of the PQ-tree for interval graphs that if the subtree rooted at the child~$y_i$ contains a leaf that corresponds to a maximal clique of $G$ that contains~$v$, then all leaves in this subtree correspond to maximal cliques containing~$v$~\cite{Korte1989}.
    Consider the set of links $\{x,y_i\}$
    where all leaves of the subtree rooted at $y_i$ are maximal cliques containing~$v$.
    Since $G$ is an interval graph, the leaves containing~$v$ are consecutive in~$T$.  Hence, the links in this set are consecutive and form a segment of~$x$.  We call this segment the \emph{segment of~$v$} and denote it by~\emph{$S_v$}.
    In~\Cref{fig:MPQ-tree-example-b}, for example, we have $S_1 = \{x_1x_2, x_1C_4\}$ and $S_2 = \{x_2C_1, x_2C_2\}$.
    For~$1 \le i \le k$, we let \emph{$L_i$} $\subseteq W_x$ denote the set of vertices that are assigned to $x$ and whose segment contains~$\{x,y_i\}$. In~\Cref{fig:MPQ-tree-example-b}, for example, we have $L_1 = \{1\}$, $L_2 = \{1, 6\}$, $L_3 = \{6, 7\}$ and $L_4 = \{7\}$ for~$x_1$.
    The maximal clique represented by a leaf $x$ of~$T$ is the set $W_x \cup W_{P} \cup W_{Q}$, where $W_{P}$ denotes the set containing all vertices of~$G$ that are assigned to P-nodes on the path $p$ from $x$ to the root, and $W_{Q}$ denotes the set containing the vertices in $L_i$ for every edge $qy_i$ in $p$ where $y_i$ is a child of a Q-node $q$.
    
    We say that a vertex $v$ of $G$ is \emph{above} a node $x$ if $v$ is assigned to a node on the path from $x$ to the root in $T$, excluding~$x$.
    Similarly, we say that $v$ is \emph{below} $x$ if $v$ is assigned to a node on a path from $x$ to some leaf in the subtree $T_x$ of $T$ rooted at~$x$, excluding $x$.
    The set of all vertices of $G$ that are above node $x$ is referred to as $A_x^T$ and the set of all vertices below $x$ as $B_x^T$.
    In~\Cref{fig:MPQ-tree-example-b}, for example, we have $A_{x_2}^T = \{1, 6, 7\}$ and $B_{x_2}^T = \{3, 5\}$.

    Since $T$ is an ordered tree, the leaves of $T$ are ordered.
    We say that an interval representation \I of $G$ \emph{agrees} with $T$ if the order of the maximal cliques of $G$ in \I coincides with the order of the maximal cliques described by the leaves of $T$.
    We let \emph{$\pi_x^T$} denote the order of the children of $x$ in $T$.
    A \emph{rotation} of an MPQ-tree $T$ is an MPQ-tree $T'$ obtained from $T$ by arbitrarily permuting the order of the children of P-nodes and reversing the order of the children of any subset of Q-nodes in $T$.
    We say that an MPQ-tree $T$ is \emph{$\varphi$-orientable} if it agrees with some $\varphi$-orientable interval representation of $G$.
%
    \noindent Let $a$ be an arc in $G$ with endpoints $u$ and $v$.
    If $u$ and $v$ are assigned to the same node of~$T$, then we call $a$ \emph{horizontal}.
    Otherwise, $a$ is \emph{vertical}, and $u$ and $v$ are assigned to distinct nodes $x_u$ and $x_v$ of~$T$, respectively.
    In this case, if $x_u$ is an ancestor of $x_v$ in $T$, then $u$ is called the \emph{upper} and $v$ the \emph{lower endpoint} of $a$ .
    
    We show that certain vertical and horizontal arcs impose restrictions on the clique order, and thus 
    on the order of the children of some nodes in $T$.
    Let $x$ be a node of $T$ with children $y_1, \dots, y_k$ and let $B_i = B_{y_i}^T \cup W_{y_i}$ with $i \in \{1, \dots, k\}$.
    We say that $y_i$ is a \emph{constraining child} of $x$ if $x$ is a P-node and there exists a vertical arc $a$ in $G$ with the lower endpoint in $B_i$ and the upper endpoint in $W_x \cup A_x^T$.
    We say that $x$ is \emph{constrained} if one of the following holds:
    \begin{itemize}
        \item $x$ is a P-node and $x$ has a constraining child $y_i$ or
        \item $x$ is a Q-node and there is an arc~$a$ with endpoints~$u$ and~$v$ such that
        \begin{description}
            \item[(V1)] $a$ is vertical, $u \in B_x^T$, and $v \in A_x^T$ or 
            \item[(V2)] $a$ is vertical, $u \in W_x$, $S_u \neq S_{1,k}$, and $v \in A_x^T$ or 
            \item[(V3)] $a$ is vertical, $u \in B_x^T$, and $v \in W_x$ or
            \item[(H)] \phantom{$a$}$a$ is horizontal, $u,v \in W_x$, and $S_u \neq S_v$.
        \end{description}
    \end{itemize}
    \noindent \Cref{fig:constr-children} illustrates constrained nodes.
    \begin{figure}
        \centering
        \textsfbf{(a)}\ \includegraphics[page=1,valign=t]{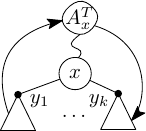} \hfil
        \textsfbf{(b)}\ \includegraphics[page=2,valign=t]{fig/constrained-children} \hfil
        \textsfbf{(c)}\ \includegraphics[page=3,valign=t]{fig/constrained-children}
        \caption{(a) Constrained P-node with constraining children $y_1$ and $y_k$; (b)~examples for the three types of vertical arcs constraining a Q-node~$x$; (c)~a horizontal arc constraining a Q-node~$x$.}
        \label{fig:constr-children}
    \end{figure}
    We say that the position of a child $y_i$ of $x$ is \emph{fixed}, if~$y_i$ has the same position in $\pi_x^T$ for every $\varphi$-oriented MPQ-tree $T$ of $G$.
    We say that $\pi_x^T$ is fixed if it is the same for every $\varphi$-orientable rotation $T$ of the MPQ-tree of $G$.
    Next, we show that the pair~$(\varphi, \sigma)$ is consistent, that is, $G$ admits a $\varphi$-orientable interval representation with clique order $\sigma$, if and only if the rotation~$T$ of the MPQ-tree of $G$ corresponding to~$\sigma$ satisfies the following constraints.
    \begin{enumerate}[(a)]
        \item If $x$ is a constrained P-node, then the position of each constraining child $y_i$ in $\pi_x^T$ is fixed and must be either the first or the last child of $x$.
        \item If $x$ is a constrained Q-node, then $\pi_x^T$ is fixed, i.e., for every $\varphi$-orientable rotation~$T'$ of~$T$, we have $\pi_x^{T'} = \pi_x^T$.
    \end{enumerate}
    We say that a clique order~$\sigma$ of~$G$ is~\emph{$\varphi$-consistent} if it is the clique order defined by a rotation~$T$ of the MPQ-tree of $G$ that satisfies properties (a) and~(b).
    We give a proof sketch for the following lemma, which is the first part of the main result of this section.

    \begin{restatable}{lemma}{lemcosistentorder}
        \restateref{lem:consistent-order}    
        \label{lem:consistent-order}    
        Let~$G$ be a mixed graph, let~$\varphi$ be an orientation of $G$ and let~$\sigma$ be a~$\varphi$-consistent clique order of $G$.  Then~$G$ admits a~$\varphi$-orientable interval representation if and only if it admits one where the clique order is~$\sigma$.
    \end{restatable}
    \begin{proof}[Sketch of proof.]
    Recall that a clique order $\sigma$ of $G$ corresponds to the leaf order of an MPQ-tree~$T$ of~$G$.
    Assume $G$ admits a $\varphi$-orientable interval representation \I and let $\sigma$ be a clique order of $G$.
    We show that $G$ admits a $\varphi$-orientable interval representation with clique order $\sigma$ if and only if $\sigma$ is~$\varphi$-consistent.
     We give an example for the necessity; the full proof is shown in~\Cref{lem:constrained-child} in the appendix.
     Let $x$ be a P-node with constraining child $y_i$ and assume there exists a vertical arc $(u,v)$ in $G$ such that $u \in B_i$ and $v \in W_x \cup A_x^T$.
        Observe that $y_i$ must be the first child of $x$ if $\varphi(u) = \varphi(v) = 1$:
        If, for the sake of contradiction, some~$y_j$,~$i \neq j$ is the first child of $x$,
        then there exists a vertex $w \in B_j$ that is not adjacent to~$u$.
        Consider the intervals $I_u = [l_{u}, r_{u}]$, $I_{v} = [l_{v}, r_{v}]$, and $I_{w} = [l_{w}, r_{w}]$ in \I.
        We have~$l_{u} < l_{v}$, as $(u,v)$ is an arc in $G$ and $\varphi(u) = \varphi(v) = 1$, and $r_{w} < l_{u}$, as $T$ puts $y_j$ before $y_i$; see \cref{fig:wrong-child-rotation}.
        Therefore,~$r_{w} < l_{v}$ and thus $I_{v}$ and $I_{w}$ are disjoint.
        This is a contradiction to the properties of the MPQ-tree, by which every vertex in $W_x \cup A_x^T$ must be adjacent to every vertex in $B_x^T$.
        We show similar statements for Q-nodes in the appendix.
     
        \begin{figure}[htb] 
            \centering
            \begin{subcaptiongroup}
            \textsfbf{(a)} \ \includegraphics[page=1,valign=t]{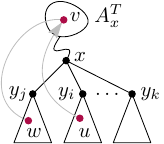} \hfil
            \textsfbf{(b)} \quad \includegraphics[page=2,valign=t]{fig/wrong_child_rotation-1.pdf}
            \end{subcaptiongroup}
            \caption{(a) MPQ-tree where $y_i$ is not the last child of $x$;
                     (b)~the implied interval representation.}
            \label{fig:wrong-child-rotation}
        \end{figure}

    
    For the converse, let $T$ be an MPQ-tree of $G$ that satisfies the constraints.
    Let $T'$ be the rotation of the MPQ-tree of $G$ that agrees with \I.
    We choose \I such that $T'$ is as similar to~$T$ as possible, meaning that the number of nodes for which the order of their children is the same in both trees is maximized.
    If $T = T'$ we are done.
    Otherwise, let $x$ be the lowest node in~$T$ for which $\pi_x^T \neq \pi_x^{T'}$.
    Let $T''$ be the MPQ-tree we obtain from $T'$ by changing $\pi_x^{T'}$ to $\pi_x^T$.
    Observe that for each node $x$ the set of nodes assigned to~$x$ as well as the sets~$A_x^T$ and~$B_x^T$ are the same for all three trees~$T,T',T''$.
    We claim that $T''$ agrees with a $\varphi$-oriented interval representation $\I'$ of $G$, which contradicts the choice of $T'$, and hence $T = T'$ follows.

    The full proof considers two cases, namely, $x$ is a P-node and $x$ is a Q-node.
    In both cases, we use the fact that there exists a \emph{normalized} interval representation $\I'$ in which the following properties hold for all nodes $x$ in $T$:
    \begin{itemize}
        \item intervals that contain each other in \I also contain each other in $\I'$ and
        \item if $x$ is a P-node, then for every non-constraining child $y_i$ of $x$, every edge between $B_i$ and~$W_x \cup A_x^T$ is represented by containment in $\I'$ or
        \item if $x$ is an unconstrained Q-node, then every edge between (i) $B_x^T$ and~$A_x^T$, (ii) $B_x^T$ and $u \in W_x$ with $S_u = S_{1,k}$ and (iii) $u \in W_x \cup A_x^T$ and $v \in W_x$ with $S_v \neq S_{1,k}$ is represented by containment in $\I'$.
    \end{itemize}
    The existence of a normalized representation is shown in~\Cref{cor:constr-cont} in the appendix.
    For a P-node $x$ with children $y_1,\dots, y_k$, we then show the following.
    Property (a) gives us that~$\pi_x^T$ and $\pi_x^{T''}$ agree on the positions of the constraining children.
    In other words, $\pi_x^T$ and $\pi_x^{T''}$ differ only with respect to the order of the non-constraining children of $x$.
    Observe that $\pi_x^{T''}$ can be obtained from $\pi_x^T$ by repeatedly switching the order of two consecutive non-constraining children of $x$.
    In order to show that there exists a $\varphi$-oriented interval representation that agrees with $T''$, we show the following claim in the appendix.
    \begin{claim*}
        Let $T^\star$ be the rotation of $T'$ we get by switching the order of two consecutive non-constraining children $y_i, y_{i+1}$ of $x$ in $T'$.
        There exists a $\varphi$-orientable interval representation~$\I'$ that agrees with $T^\star$.
    \end{claim*}

    The idea is that since $x$ is not constrained, each interval of $\I'[B_i]$ and $\I'[B_{i+1}]$ is properly contained in every interval of $\I'[W_x \cup A_x^{T'}]$.
    We can therefore apply \cref{lem:rev-rep} twice---first with $W = \cont_\mathcal{I'}(I)$ where $I$ is the smallest interval that spans $\I'[B_i] \cup \I'[B_{i+1}]$ and then with $W_i = \cont_\mathcal{I''}(I_i)$ and $W_{i+1} = \cont_\mathcal{I''}(I_{i+1})$ where $I_i$ and $I_{i+1}$ are the smallest intervals that span $\I''[B_i]$ and $\I''[B_{i+1}]$, respectively, where 
    $\I'' = \mirr_I(\I')$---to obtain a $\varphi$-orientable interval
    representation that agrees with $T^\star$.
    That is, the order of $\I'[B_i]$ and $\I'[B_{i+1}]$ coincides with the order of $y_i$ and $y_{i+1}$ in $T^\star$. 
    %
%
%
    By repeatedly switching consecutive non-constraining children of $x$, we can therefore construct a $\varphi$-orientable interval representation that agrees with $T''$.
    This contradicts our choice of $T'$.

    For a Q-node $x$, observe that since $\pi_x^T$ and $\pi_x^{T'}$ differ, it follows from Property (b) that $x$ is not constrained, since otherwise $\pi_x^T$ would be fixed and therefore the same in $T$ and $T'$.
    Let $U \subseteq W_x$ be the set of vertices $u$ of $G$ that are assigned to $x$ and for which $S_u \neq S_{1,k}$.
    To show that $T$ agrees with a $\varphi$-orientable interval representation, first show that (i) $G[U]$ is $\varphi$-orientable if and only if it is $\mathrm{rev}_U(\varphi)$-orientable and that (ii) there are only (undirected) edges between $B_x^T \cup U$ and $A_x^T \cup (W_x \setminus U)$.
    We then apply \cref{lem:rev-orient} to obtain a $\mathrm{rev}_{\varphi}$-orientable interval representation $\I''$.
    Altogether, we show that there exists a $\varphi$-orientable interval representation of $G$ that agrees with $T''$, a contradiction to our choice of $T'$.
    \end{proof}

    We remark that the MPQ-tree $T$ and the constrained children of each node can be computed in polynomial time.  Therefore, in polynomial time we can decide whether a mixed graph admits a $\varphi$-consistent clique order and compute one in the positive case.  In fact, by~\Cref{lem:mpq-tree-alg}, which we prove in the appendix, this can be achieved in linear time.

    \newcommand{\lc}{\mathrm{lc}}
    \newcommand{\rc}{\mathrm{rc}}
    
    This result shows that, given an orientation~$\varphi$ of~$G$, we can construct an MPQ-tree that represents every clique order~$\sigma$ of~$G$ such that the pair~$(\varphi, \sigma)$ is consistent if and only if $G$ is an oriented interval graph.
    It remains to construct a representation from~$(\varphi, \sigma)$.

     Recall the definition of augmented graph $G^+$ from \cref{subsec:constructing-phi} and let $C = \{c_1,\dots, c_k\}$ be the set of clique point vertices.
    Let $V^L = \{v \in V(G) \mid \varphi(v) = -1\}$,
    let $G^L = G^+[V^L \cup C]$, and let
    $\sigma^L = C_1^L,\dots, C_k^L$, where, for $i \in \{1,\dots,k\}$,
    $C_i^L = C_i^+ \cap V(G^L)$; see
    \cref{fig:L-R-split-overview-c}.  Define $V^R$, $G^R$,
    $C_1^R,\dots,C_k^R$, and $\sigma^R$ analogously; see
    \cref{fig:L-R-split-overview-e}.
    Let $\varphi^L$ and $\varphi^R$ denote the orientations
    of~$G^L$ and $G^R$ that assign to each vertex the
    orientation~$-1$ and~$+1$, respectively.
    We show the following statement.

    \begin{restatable}{lemma}{lemlrsplit}
    \restateref{lem:L-R-split}
    \label{lem:L-R-split}
      A mixed graph~$G$ admits a $\varphi$-orientable interval
      representation with clique order~$\sigma$ if and only if
      \begin{itemize}
      \item the graph~$G^L$ admits a $\varphi^L$-orientable interval
        representation with clique order~$\sigma^L$ and        
      \item the graph~$G^R$ admits a $\varphi^R$-orientable interval
        representation with clique order~$\sigma^R$.
      \end{itemize}
      Given representations of~$G^L$ and~$G^R$, a
      $\varphi$-orientable interval representation of~$G$ can be
      computed from representations of~$G^L$ and~$G^R$ in linear time.
    \end{restatable}

            

        
        

      
      
    \begin{proof}
      ``$\Rightarrow$'': We take the $\varphi$-orientable interval
      representation~\I of~$G$ with clique order~$\sigma$ and split it
      into two representations, a representation~$\I^L$ for the
      left-oriented intervals and a representation~$\I^R$ for the
      right-oriented intervals. To both of these, we add, for each
      clique points in~$C$, a very short interval that intersects all
      intervals that form the corresponding clique. The first
      resulting representation is $\varphi^L$-oriented and has clique
      order~$\sigma^L$, the second one is $\varphi^R$-oriented and has
      clique order~$\sigma^R$, as desired.

      ``$\Leftarrow$'': Let~$\I^L$ be the $\varphi^L$-orientable
      interval representation of~$G^L-C$, and let~$\I^R$ be the
      $\varphi^R$-orientable interval representation~$\I^R$ of~$G^R-C$.
      We join~$\I^L$ and~$\I^R$ as follows.  Since~$\sigma^L$ and~$\sigma^R$ 
      both have a clique for each clique point vertex, we
      can simply set, for each $i \in \{1,\dots,k\}$,
      $C_i = (C_i^L \cup C_i^R) \setminus \{c_i\}$ and
      $\sigma = C_1,\dots,C_k$.  It remains to determine, for each
      $i \in \{1,\dots,k\}$, the order of the interval endpoints within~$C_i$.
      To this end, we first
      place the left endpoints of the intervals in~$V(G)$ that start
      in~$C_i$ and extend to the right; then we place the right
      endpoints of the intervals in~$V(G)$ that start in~$C_i$ and
      extend to the left.  We order the left endpoints as follows: we
      start with the intervals in $V^L$ in the order in which they
      appear in~$\I^L$; then we add the intervals in $V^R$ in the
      order in which they appear in~$\I^R$.  We order the right
      endpoints accordingly. This fixes a set~\I of intervals
      that contains an interval for each vertex of~$G$;
      it can be computed in linear time by merging~$\I^L$ and~$\I^R$.

      Note that two intervals in~\I intersect if and only if they are both contained in a maximal clique~$C_i$. This ensures
      that all undirected edges of~$G$ are realized in~\I.  Every arc
      $(u,v)$ of $G$ is realized in~\I because either both~$u$ and~$v$ lie
      in $V^L \setminus C$ and $(u,v)$ is realized in~$\I^L$, or they
      both lie in $V^R \setminus C$ and $(u,v)$ is realized in~$\I^R$.
      Since the order of the endpoints of intervals in
      $V^L \setminus C$ (in $V^R \setminus C$) is the same in~$\I^L$
      (in~$\I^R$) and in $\I$, the arc $(u,v)$ is realized in~\I.
    \end{proof}

    We are now ready to prove the main result of this section.

\thmgivenphigetrep*

    \begin{proof}
      We use \cref{lem:mpq-tree-alg} to compute, in linear time, a clique
      order~$\sigma$ such that $G$ admits a $\varphi$-orientable
      interval representation if and only if $G$ admits such a
      representation with clique order~$\sigma$. Then we split
      $(G,\sigma)$ into $(G^L,\sigma^L)$ and $(G^R,\sigma^R)$ in linear time.
      Since all vertices of~$G^L$ have the same orientation, the
      undirected edges of~$G^L$ are precisely the containment edges,
      and we can use \cref{cor:linear-time-given-sigma} to check
      in linear time whether $G^L$ admits a
      $\varphi^L$-orientable interval representation with clique
      order~$\sigma^L$.  Then we do the
      corresponding check for~$G^R$.  If both checks are successful,
      we use \cref{lem:L-R-split} to combine the
      resulting interval representations of~$G^L$ and~$G^R$ in linear time to a
      $\varphi$-orientable interval representation of~$G$.
    \end{proof}

\begin{toappendix}
    \section{Prescribed Orientations (Full Version)}
    \label{sec:given-phi}

    In this section, we assume that $\varphi$ is given as part of the input.
    We present an algorithm that recognizes $\varphi$-orientable interval graphs using clique orderings.
    The proof works in two steps.
    The first step is to show that we can compute a clique ordering $\sigma$ of $G$ such that $G$ has a $\varphi$-oriented
    interval representation \I if and only if it has a $\varphi$-oriented interval representation with
    clique ordering $\sigma$.
    In the second step we construct a $\varphi$-orientable interval representation of
    $G$ with clique ordering $\sigma$ if and only if $G$ is a $\varphi$-oriented interval graph.
    To describe the clique orderings of $G$ we use modified \emph{PQ-trees}~\cite{Korte1989}.
    An \emph{MPQ-tree} $T$ of an interval graph $G$ is a rooted, ordered tree whose nodes are called \emph{P-} and
    \emph{Q-nodes}.
    The nodes in $T$ are connected via \emph{links}.
    Let $x$ be a node in $T$ and let $\{y_1, \dots, y_k\}$ be the children of $x$.
    For a Q-node $x$, a set of consecutive links $S_{i,j} = \{xy_i, \dots, xy_j\}$ is called a \emph{segment}
    of $x$.
    Similar to PQ-trees, the leaves of the MPQ-tree $T$ of $G$ are the maximal cliques of $G$.
    We assign every vertex $v$ in $G$ to the node $x$ in $T$ that is the lowest common ancestor of all cliques containing~$v$. 
    Let~\emph{$W_x$} denote the set of vertices assigned to~$x$.
    If $x$ is a Q-node, then consider the set of links $\{x,y_i\}$
    where the leaves of the subtree rooted at $y_i$ are maximal
    cliques containing~$v$.
    Since $G$ is an interval graph, the leaves containing $v$ are consecutive in $T$, hence the links  in this set are consecutive and form a segment of $x$.
    Let~\emph{$S_v$} denote the segment of $v$.
        For a child $y_i$ of $x$, consider the leaves $c_1, \dots, c_l$ of the subtree rooted at $y_i$.
    Let \emph{$L_i$} $\subset W_x$ denote the set of vertices that are assigned to $x$ and that are contained in the maximal cliques that correspond to the leaves $c_1, \dots, c_l$.
    We say that a vertex $v$ of $G$ is \emph{above} a node $x$ if $v$ is assigned to a node on the path from $x$ to the root in $T$, excluding $x$.
    Similarly, we say that $v$ is \emph{below} $x$ if $v$ is assigned to a node on a path from $x$ to some leaf in the subtree $T_x$ of $T$ rooted at~$x$, excluding $x$.
    The set of all vertices of $G$ that are above node $x$ is referred to as $A_x^T$ and the set of all
    vertices below $x$ as $B_x^T$.
    The maximal cliques of $G$ are precisely the sets $W_x \cup A_x^T$ where $x$ is a leaf of $T$.
    Since $T$ is an ordered tree, the leaves of $T$ are ordered.
    We say that an interval representation \I of $G$ \emph{agrees} with $T$ if the order of the maximal
    cliques of $G$ in \I coincides with the order of the maximal cliques described by the leaves of $T$.
    We let \emph{$\pi_x^T$} denote the order of the children of $x$ in $T$.
    A \emph{rotation} of an MPQ-tree $T$ is an MPQ-tree $T'$ obtained from $T$ by arbitrarily permuting the order of the children of P-nodes and reversing the order of the children of any subset of Q-nodes in $T$.
    We say that an MPQ-tree $T$ is \emph{$\varphi$-orientable} if it agrees with some $\varphi$-orientable interval representation of $G$.
    We assume the following properties of MPQ-trees.
    \begin{proposition}[see \cite{Korte1989}, Lemma 2.2]
      \label{obs:p-q-nodes}
      Let $x$ be a non-leaf node of an MPQ-tree, and let
      $y_1, \dots,y_k$ be the children of~$x$.
      If $x$ is a P-node, then, for every $i \in \{1, \dots, k\}$,
      there is at least one vertex assigned to~$y_i$ or to a node in
      the subtree rooted at $y_i$, i.e.,
      $W_{y_i} \cup B^T_{y_i} \ne \emptyset$.
      If $x$ is a Q-node, then $k \geq 3$ and
      \begin{enumerate}[(a)]
            \item the subtrees rooted at $y_1$ and $y_k$ are non-empty,
            \item $L_1 \subsetneq L_2$ and $L_k \subsetneq L_{k-1}$, and,
            \item for every $i \in \{2, \dots, k-1\}$, it holds that
              $(L_i \cap L_{i+1}) \setminus L_1 \neq \emptyset$ and
              $(L_{i-1} \cap L_i) \setminus L_k \neq \emptyset$.
      \end{enumerate}
    \end{proposition}
    
    \noindent Let $a$ be an arc in $G$ with endpoints $u$ and $v$.
    If $u$ and $v$ are assigned to the same node of~$T$, then we call $a$ \emph{horizontal}.
    Otherwise, $a$ is \emph{vertical}, and $u$ and $v$ are assigned to distinct nodes $x_u$ and $x_v$ of~$T$, respectively.
    In this case, if $x_u$ is an ancestor of $x_v$ in $T$, then $u$ is called the \emph{upper endpoint} and $v$ is called the \emph{lower endpoint} of $a$ .
    Let $x$ be a node of $T$ with children $y_1, \dots, y_k$ and let $B_i = B_{y_i}^T \cup W_{y_i}$ with $i \in \{1, \dots, k\}$.
    We say that $y_i$ is a \emph{constraining child} of $x$ if $x$ is a P-node and there exists a vertical arc $a$ in $G$ with the lower endpoint in $B_i$ and the upper endpoint in $W_x \cup A_x^T$.
    We say that a node $x$ is \emph{constrained} if one of the following holds:
    \begin{itemize}
        \item $x$ is a P-node and $x$ has a constraining child $y_i$ or
        \item $x$ is a Q-node and there is an arc~$a$ with endpoints~$u$ and~$v$ such that
        \begin{description}
            \item[(V1)] $a$ is vertical, $u \in B_x^T$, and $v \in A_x^T$ or 
            \item[(V2)] $a$ is vertical, $u \in W_x$, $S_u \neq S_{1,k}$, and $v \in A_x^T$ or 
            \item[(V3)] $a$ is vertical, $u \in B_x^T$, and $v \in W_x$ or
            \item[(H)] $a$ is horizontal, $u,v \in W_x$, and $S_u \neq S_v$.
        \end{description}
    \end{itemize}
    \noindent \Cref{fig:constr-children} illustrates constrained nodes.
    
    We say that the position of a child $y_i$ of $x$ is \emph{fixed} if $y_i$ has the same position in $\pi_T^x$ for every $\varphi$-oriented rotation $T$ of the MPQ-tree of $G$.
    We say that $\pi_x^T$ is fixed if it is the same for every $\varphi$-oriented MPQ-tree $T$ of $G$.
    We show that a node $x$ in $T$ being constrained implies necessary properties for $\pi_x^T$.

    \begin{restatable}{lemma}{lemconstrchild}
        \restateref{lem:constrained-child}
        \label{lem:constrained-child}
        Let $T$ be a $\varphi$-orientable MPQ-tree of a mixed interval graph $G$ and let $x$ be a node in $T$ with children $\{y_1, \dots, y_k\}$.
        \begin{enumerate}[(a)]
            \item If $x$ is a constrained P-node, then the position of each constraining child $y_i$ in $\pi_x^T$ is fixed and must be either the first or the last child of $x$.
            \item If $x$ is a constrained Q-node, then $\pi_x^T$ is fixed, i.e., for every $\varphi$-orientable rotation~$T'$ of~$T$, we have $\pi_x^{T'} = \pi_x^T$.
        \end{enumerate}
    \end{restatable}

    \begin{proof}
        Let \I be a $\varphi$-orientable interval representation of $G$ that agrees with $T$ and
        let $x$ be a constrained node in $T$. 
        Let $y_1 \prec \cdots \prec y_k$ be the order of the children of $x$.
        We consider the two types of nodes in $T$ and the ways in which they can be constrained.

        Case 1: $x$ is a constrained P-node. 
        The properties of MPQ-trees give us that (i) every vertex in $W_x \cup A_x^T$ is adjacent to every vertex in $B_x^T$, (ii) none of the $B_i$ are empty, and (iii) any two vertices $b_i \in B_i$, $b_j \in B_j$ with $i \neq j$ are non-adjacent in~$G$.
        Let $y_i$ be a constraining child of $x$, that is, there exists a vertical arc $a$ with endpoints $u$ and $v$ such that $u \in B_i$ and $v \in W_x \cup A_x^T$.
        Assume $(u,v) \in \Arcs(G)$.
        We show that $y_i$ must be the first child of $x$ if $\varphi(u) = \varphi(v) = 1$.
        Assume for the sake of contradiction that some $y_j$, $i \neq j$ is the first child of $x$.
        Then by (iii), there exists a vertex $w \in B_j$ that is not adjacent to $u$.
        Consider the intervals $I_u = [l_{u}, r_{u}]$, $I_{v} = [l_{v}, r_{v}]$, and $I_{w} = [l_{w}, r_{w}]$ in \I.
        We have $l_{u} < l_{v}$, as $(u,v)$ is an arc in $G$ and $\varphi(u) = \varphi(v) = 1$, and $r_{w} < l_{u}$, as $T$ puts $y_j$ before $y_i$; see \cref{fig:wrong-child-rotation}.
        Therefore, $r_{w} < l_{v}$ and thus $I_{v}$ and $I_{w}$ are disjoint, a contradiction to (i).
        By an analogous argument, $y_i$ must be the last child of $x$ if $\varphi(u) = \varphi(v) = -1$.
        Moreover, symmetrically, if $(v,u) \in \Arcs(G)$, then 
        $y_i$ must be the last child of $x$ if $\varphi(u) = \varphi(v) = 1$ and the first if $\varphi(u) = \varphi(v) = -1$.
        Therefore, $y_i$ is fixed and must be the first or last child of $x$ if it is a constraining child.
        This proves~(a).

        Case 2: $x$ is a constrained Q-node.
        Q-nodes can be constrained by three kinds of vertical arcs (V1) -- (V3) and by horizontal arcs.
        We first consider the case where $x$ is constrained by a vertical arc.
        The types of vertical arcs can be grouped together by the way they constrain $x$.
        Namely, (V1) and (V2) arcs determines a first or last child of $x$ (similar to the P-node case), while (V3) arcs order two children of $x$ in $\pi_x^T$.
        In either case, $\pi_x^T$ is fixed.
        
        We first show that (V1) and (V2) arcs, similarly to the P-node case, determine a first or last child of $x$.
        Let $a$ be a (V1) arc in $G$ with the lower endpoint $u \in B_x^T$ and the upper endpoint $v \in A_x^T$.
        In particular, let $y_i$ be a child of $x$ such that $xy_i$ is the first link in $S_u$.
        Assume $(u,v) \in \Arcs(G)$.
        We show that $y_i$ must be the first child of $x$ if $\varphi(u) = \varphi(v) = 1$.
        Assume for the sake of contradiction that some $y_j$, $i \neq j$ is the first child of $x$ and that $u \notin B_j$.
        From~\Cref{obs:p-q-nodes} it follows that there exists a vertex $w \in B_j$ that is adjacent to $v$ and not adjacent to $u$.
        The proof now works analogously to the P-node case.
        Namely, consider the intervals $I_u = [l_{u}, r_{u}]$, $I_{v
        } = [l_{v}, r_{v}]$, and $I_{w} = [l_{w}, r_{w}]$ in \I.
        We have $l_{u} < l_{v}$, as $(u,v)$ is an arc in $G$ and $\varphi(u) = \varphi(v) = 1$, and $r_{w} < l_{u}$, as $T$ puts $y_j$ before $y_i$; see \cref{fig:wrong-child-rotation}.
        Therefore, $r_{w} < l_{v}$ and thus $I_{v}$ and $I_{w}$ are disjoint.
        This a contradiction to $v$ and $w$ being adjacent.
        By an analogous argument, $y_i$ must be the last child of $x$ if $\varphi(u) = \varphi(v) = -1$.
        Moreover, symmetrically, if $(v,u) \in \Arcs(G)$, then $y_i$ must be the last child of $x$ if $\varphi(u) = \varphi(v) = 1$ and the first if $\varphi(u) = \varphi(v) = -1$.
        
        We now consider (V2) arcs, that is, let $a$ be an arc in $G$ with lower endpoint $u \in W_x$, $S_u \neq S_{1,k}$, and upper endpoint $v \in A_x^T$.
        Let $y_i$ be a child of $x$ such that $xy_i$ is the first link in $S_u$.
        Assume $(u,v) \in \Arcs(G)$.
        We show that $y_i$ must be the first child of $x$ if $\varphi(u) = \varphi(v) = 1$.
        Assume for the sake of contradiction that some $y_j$, $i \neq j$ is the first child of $x$ and that $u \notin L_j$, that is, the subtree rooted at $y_j$ contains no maximal clique as a leaf that contains $u$.
        From~\Cref{obs:p-q-nodes} and since $xy_j \notin S_u$, it follows that there exists a vertex $w \in B_j$ that is adjacent to $v$ and not adjacent to $u$.
        The proof now again works analogously to the P-node case.
        Namely, consider the intervals $I_u = [l_{u}, r_{u}]$, $I_{v} = [l_{v}, r_{v}]$, and $I_{w} = [l_{w}, r_{w}]$ in \I.
        We have $l_{u} < l_{v}$, as $(u,v)$ is an arc in $G$ and $\varphi(u) = \varphi(v) = 1$, and $r_{w} < l_{u}$, as $T$ puts $y_j$ before $y_i$; see \cref{fig:wrong-child-rotation}.
        Therefore, $r_{w} < l_{v}$ and thus $I_{v}$ and $I_{w}$ are disjoint, a contradiction to $v$ and $w$ being adjacent.
        By an analogous argument, $y_i$ must be the last child of $x$ if $\varphi(u) = \varphi(v) = -1$.
        Moreover, symmetrically, if $(v,u) \in \Arcs(G)$, then $y_i$ must be the last child of $x$ if $\varphi(u) = \varphi(v) = 1$ and the first if $\varphi(u) = \varphi(v) = -1$.
        
        Finally, we consider (V3) arcs.
        That is, $a$ is an arc with lower endpoint $u \in B_x^T$ and upper endpoint $v \in W_x$.
        Since $v$ is assigned to a node that is located further up in $T$ than the node $u$ is assigned to, there exists a maximal clique that contains $v$ but not $u$.
        Therefore, there exists a child $y_j$ of $x$ with $i \neq j$ such that $v$ is adjacent to some vertex $w \in B_j$ that is not adjacent to $u$.
        We show that $a$ uniquely determines $\pi_x^T$ by ordering $y_i$ and $y_j$ in $\pi_x^T$.
        Assume $(u,v) \in \Arcs(G)$.
        We consider the two possible orientations of $u$ and $v$.
        If $\varphi(u) = \varphi(v) = 1$, then $I_u$ and $I_v$ overlap in \I with $l_u < l_v < r_u < r_v$, that is, $I_u$ contains the left endpoint of $I_v$.
        Assume that $y_j \prec y_i$.
        Then, together with the fact that $u$ and $w$ are not adjacent, we have $r_w < l_u$.
        But then $r_w < l_v$, a contradiction, since $w$ is adjacent to $v$.
        Therefore, it follows that $y_i \prec y_j$; see \cref{fig:constr-q-node-below}.
        \begin{figure}
            \centering
            \includegraphics{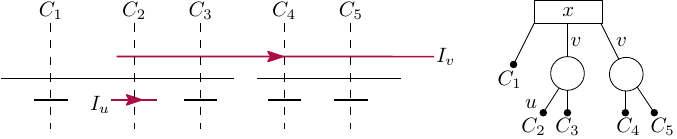}
            \caption{The arc $(u,v)$ uniquely determines the order $\pi_x^T$.}
            \label{fig:constr-q-node-below}
        \end{figure}
        Similarly, if $\varphi(u) = \varphi(v) = -1$, then we obtain $y_j \prec y_i$.
        Moreover, if $(v,u) \in \Arcs(G)$ and $\varphi(u) = \varphi(v) = 1$, then, symmetrically, we find $y_j \prec y_i$.
        Similarly, if $\varphi(u) = \varphi(v) = -1$, then we find $y_i \prec y_j$.

        It remains to consider the case where $x$ is constrained by a horizontal arc.
        Let $a$ be an arc in $G$ with endpoints $u,v \in W_x$ such that $S_u \neq S_v$. That is, $u$ is contained in $L_i$ and $v$ is contained in $L_j$ with $i \neq j$ and, without loss of generality, $u \notin L_j$.
        Assume $(u,v)\in \Arcs(G)$.
        We consider the two possible orientations of $u$ and $v$.
        If $\varphi(u) = \varphi(v) = 1$, then, since $u \notin L_j$, and $l_u < l_v < r_u < r_v$, it follows that $y_i \prec y_j$.
        Similarly, if $\varphi(u) = \varphi(v) = -1$, then we find $y_j \prec y_i$.
        Assume $(v,u)\in \Arcs(G)$.
        We analogously find $y_j \prec y_i$ for $\varphi(u) = \varphi(v) = 1$ and $y_i \prec y_j$ for $\varphi(u) = \varphi(v) = -1$.
        Altogether, this proves (b).
    \end{proof}

    This shows that constrained children imply necessary properties for an MPQ-tree $T$ to agree with an oriented interval representation of $G$:
    For P-nodes, they may fix the first and last child while for Q-nodes, they determine the unique order of the children.
    To show that these constraints are also sufficient, we first prove the following intermediate lemma.

    \begin{restatable}{lemma}{lemunconstrainedcontained}
        \restateref{lem:unconstrained-contained}
        \label{lem:unconstrained-contained}
        Let $T$ be a $\varphi$-orientable MPQ-tree of $G$ and let \I be a $\varphi$-oriented interval representation of $G$ that agrees with $T$.
        Let $x$ be a node in $T$ with children~$y_1,\dots, y_k$.
        There exists an interval representation $\I'$ of $G$ that agrees with~$T$ such that 
        \begin{itemize}
            \item intervals that contain each other in \I also contain each other in $\I'$ and
            \item if $x$ is a P-node, then for every non-constraining child $y_i$ of $x$, every edge between $B_i$ and~$W_x \cup A_x^T$ is represented by containment in $\I'$ or
            \item if $x$ is an unconstrained Q-node, then every edge between
            \begin{itemize}
                \item $B_x^T$ and~$A_x^T$,
                \item $B_x^T$ and $u \in W_x$ with $S_u = S_{1,k}$ and
                \item $v$ and $u \in W_x \cup A_x^T$ where $v \in W_x$ and $S_v \neq S_{1,k}$
            \end{itemize} 
            is represented by containment in $\I'$.
        \end{itemize}
    \end{restatable}

    \begin{proof}
        We first consider the case that $x$ is a P-node.
        Let $I_A = \bigcap_{a \in W_x \cup A_x^T} I_a$ be the intersection of all intervals in $\I[W_x \cup A_x^T]$.
        Let $l_A$ and $r_A$ denote the left and right endpoint of $I_A$, respectively.
        Observe that each interval of the interval representations $\I[B_2 \cup \dots \cup B_{k-1}]$ 
        is properly contained in any interval in $\I[W_x \cup A_x^T]$.
        If $I_A$ also spans $\I[B_1 \cup B_k]$, then we are done.
        Otherwise, assume $y_1$ is a non-constraining child of $x$ and that $\I[B_1]$ overlaps~$I_A$.
        Let $p$ be the left-most endpoint of the intervals $\{I_v \mid v \in \bigcup_{i=1}^k B_i\}$ that is contained in $I_A$.
        Let~$l_1, \dots, l_{|B_1|}$ be the left endpoints of the intervals in $\I[B_1]$ in the order in which they appear in \I from left to right.
        Let $\I'$ be the interval representation we obtain by placing $l_1, \dots, l_{|B_1|}$ in the same order as in \I inside the interval $[l_A, p]$; see \cref{fig:move-endpoints}.
        \begin{figure}
            \centering
            \begin{subcaptiongroup}
            \textsfbf{(a)} \includegraphics[page=4,valign=t]{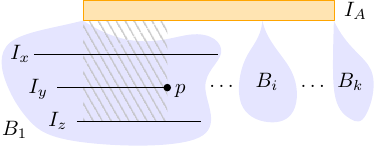} \hfil
            \textsfbf{(b)} \includegraphics[page=5,valign=t]{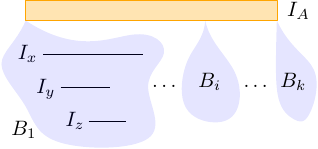}
            \end{subcaptiongroup}
            \caption{(a) The left endpoints of $I_x$, $I_y$ and $I_z$ can be placed between the left endpoint of
                $I_A$ and $p$, indicated by the gray area, to eliminate overlaps; (b) the interval representation
                $\I'$ obtained by placing the right endpoints of $I_x$, $I_y$ and $I_z$ inside $I_A$.}
            \label{fig:move-endpoints}
        \end{figure}
        We show that $\I'$ is a $\varphi$-oriented interval representation of $G$ that agrees with $T$ such that if $I_i \subset I_j$ in \I, then also $I_i \subset I_j$ in $\I'$.
        Observe that by the maximality of $p$, there are no other endpoints of intervals from \I inside $[l_A, p]$.
        Therefore, the intersections between intervals in $\I[B_1]$ and in $\I[B_i]$ for
        $i \in \{1,\dots,k\}$ as well as containments of intervals in \I are preserved in $\I'$.
        It remains to show that intervals in $\I[B_1]$ that overlap $I_A$ in \I are correctly
        represented in $\I'$.
        Let $I_b \in \I[B_1]$ be an interval that overlaps some interval $I_a \subseteq I_A$.
        Since $y_1$ is not a constraining child, there exists no arc between $I_A$ and $\I[B_1]$.
        Hence, $\{a,b\}$ is an edge and is represented correctly by $I_b \subset I_a$ in $\I'$.
        This holds for any $b \in B_1$ and $a \in A_x^T$, therefore, the adjacencies between $B_1$ and $A_x^T$ are
        correctly represented by containment in $\I'$.

        If $y_k$ is an unconstrained child and $\I[B_k]$ overlaps~$I_A$, the proof works analogously.
        Namely, we obtain the interval representation $\I'$ by placing the right endpoints of the intervals
        in $\I[B_k]$ inside $[p,r_A]$ in the same order as in \I, where $p$ is the right-most
        endpoint of the intervals $\{I_v \mid v \in \bigcup_{i=1}^k B_i\}$ contained in~$I_A$.
        Using the same arguments as above, we again get that $\I'$ is a $\varphi$-oriented interval
        representation of $G$ that agrees with $T$.

        Now consider the case that $x$ is a Q-node.
        Similar to the P-node case, we consider an interval that spans the set of intervals we want to be contained.
        Since $W_x$ may contain vertices $v$ with $S_v \neq S_{1,k}$, we cannot just form the intersection of all intervals in $W_x \cup A_x^T$, as we did in the P-node case.
        Instead, let $U = \{u \mid u \in W_x \text{ and } S_u = S_{1,k}\}$ be the vertices of $G$ that are assigned to $x$ and adjacent to every vertex in $B_x^T$ and $W_x$. 
        Further, let $I_A = \bigcap_{a \in U \cup A_x^T} I_a$ be the intersection of all intervals in $\I[U \cup A_x^T]$, an interval representation induced by vertices that are all adjacent to every vertex in $B_x^T \cup (W_x \setminus U)$.
        Similar to the P-node case, we identify a set of vertices whose intervals in \I are properly contained in $I_A$.
        For a Q-node, this set consists of $B_2 \cup \dots \cup B_{k-2}$ as well as a subset of $L_2 \cup \dots \cup L_{k-1}$: a vertex $v$ can be contained in both $L_1$ and $L_2$ and its interval $I_v$ might overlap $I_A$.  We therefore define~$B_i' = (B_i \cup L_i) \setminus U$ and consider the set $B = \bigcup_{i=2}^{k-1} B_i' \setminus (B_1 \cup B_k \cup L_1 \cup L_k)$. 
        Since the vertices in $U \cup A_x^T$ are all adjacent to every vertex in $B_x^T \cup (W_x \setminus U)$ and we removed all vertices whose intervals might overlap $I_A$, $I_A$ spans $\I[B]$.
        If $I_A$ also spans $\I[B_1' \cup B_k']$, then we are done.
        Otherwise, assume $x$ is not constrained and that $\I[B_1']$ overlaps $I_A$.
        We obtain $\I'$ analogous to the P-node case.
        That is, we place the left endpoints $l_1, \dots, l_{|B_1'|}$ of intervals in $\I[B_1']$  in the same order as in \I inside the interval $[l_A, p]$ where $p$ is the left-most endpoint of the intervals $\{I_v \mid v \in \bigcup_{i=1}^k B_i'\}$ that is contained in $I_A$.
        We show, similar to the P-node case, that $\I'$ is a $\varphi$-oriented interval representation of $G$ that agrees with $T$ such that if $I_i \subset I_j$ in \I, then also $I_i \subset I_j$ in $\I'$.
        By the same argument as above, namely, the maximality of $p$, the intersections between intervals in $\I[B_1']$ and in $\I[B_i']$ for
        $i \in \{1,\dots,k\}$ as well as containments of intervals in \I are preserved in $\I'$.
        We show that intervals in $\I[B_1']$ that overlap $I_A$ in \I are correctly
        represented in $\I'$.
        Let $I_b \in \I[B_1']$ be an interval that overlaps some interval $I_a \subseteq I_A$.
        Since $x$ is not constrained, there exists no (V1) arc between $B_x^T$ and $A_x^T$, no (V2) arc between $L_1 \setminus U$ and $A_x^T$, no (V3) arc between $B_x^T$ and $W_x$ and no (H) arc between $L_1 \setminus U$ and $U$, as that would be a constraining horizontal edge.
        Therefore, there are no arcs between $I_A$ and $\I[B_1']$.
        Hence, $\{a,b\}$ is an edge and is represented correctly by $I_b \subset I_a$ in $\I'$.
        This holds for any $b \in B_1'$ and $a \in U \cup A_x^T$, therefore, the adjacencies between $B_1'$ and $A_x^T$ are correctly represented by containment in $\I'$.

        If $\I[B_k']$ overlaps~$I_A$, the proof works analogously.
        Namely, we obtain the interval representation $\I'$ by placing the right endpoints of the intervals
        in $\I[B_k']$ inside $[p,r_A]$ in the same order as in \I, where $p$ is the right-most
        endpoint of the intervals $\{I_v \mid v \in \bigcup_{i=1}^k B_i'\}$ contained in~$I_A$.
        Using the same arguments as above, we again get that $\I'$ is a $\varphi$-oriented interval
        representation of $G$ that agrees with $T$.
    \end{proof}

    By repeatedly applying \cref{lem:unconstrained-contained} to nodes of an MPQ-tree with constrained children, we obtain a representation where all nodes satisfy the properties expressed in \cref{lem:unconstrained-contained}.  We call such a representation \emph{normalized}.

    \begin{corollary}
        \label{cor:constr-cont}
        Let $T$ be a $\varphi$-orientable MPQ-tree of $G$ and let \I be a $\varphi$-orientable interval
        representation of $G$ that agrees with $T$.
        Then there exists~$\varphi$-orientable interval representation~$\I'$ of $G$ that agrees with~$T$ and that is normalized.
    \end{corollary}

    Note that while this allows to require some of the edges of $G$ to be inclusion edges, non-arcs to vertices below a constrained child are not covered. Missing a viable full set of inclusion edges, we can thus not just apply \cref{thm:recognize-econt}. Instead we use \cref{cor:constr-cont} for the following lemma. 

    A clique order~$\sigma$ of~$G$ is~\emph{$\varphi$-consistent} if it is the clique order defined by an MPQ-tree $T$ of $G$ that has the properties stated in~\Cref{lem:constrained-child}.

    \whenappendix{\lemcosistentorder*\label{lem:consistent-order*}}

    \begin{proof}
        Recall that a clique order $\sigma$ of $G$ corresponds to the leaf order of an MPQ-tree~$T$ of~$G$.
        Assume $G$ admits a $\varphi$-orientable interval representation \I and let $\sigma$ be a clique order of $G$.
        We show that $G$ admits a $\varphi$-orientable interval representation with clique order $\sigma$ if and only if $\sigma$ is~$\varphi$-consistent, i.e., the corresponding MPQ-tree $T$ of $G$ satisfies the constraints of \cref{lem:constrained-child}.
        The necessity follows from \cref{lem:constrained-child}.
        For the converse, let $T$ be an MPQ-tree of $G$ that satisfies the constraints.
        Let $T'$ be the rotation of the MPQ-tree of $G$ that agrees with \I.
        We choose \I such that $T'$ is as similar to $T$ as possible, meaning that the number of nodes for which the order of their children is the same in both trees is maximized.
        If $T = T'$ we are done.
        Otherwise, let $x$ be the lowest node in~$T$ for which $\pi_x^T \neq \pi_x^{T'}$.
        Let $T''$ be the MPQ-tree we obtain from $T'$ by changing $\pi_x^{T'}$ to $\pi_x^T$.
        Observe that for each node $x$ the set of nodes assigned to~$x$ as well as the sets~$A_x^T$ and~$B_x^T$ are the same for all three trees~$T,T',T''$.
        We claim that $T''$ agrees with a $\varphi$-oriented interval representation $\I'$ of $G$, which is a contradiction to the choice of $T'$, and hence $T = T'$ follows.

        Let $\I'$ be the normalized interval representation of $G$ we get from~\Cref{cor:constr-cont}.
        First, consider the case where $x$ is a P-node with children $y_1,\dots, y_k$.
        \Cref{lem:constrained-child} gives us that $\pi_x^T$ and $\pi_x^{T''}$ agree on the positions of the constraining children.
        In other words, $\pi_x^T$ and $\pi_x^{T''}$ differ only with respect to the order of the non-constraining children of $x$.
        Observe that $\pi_x^{T''}$ can be obtained from $\pi_x^T$ by repeatedly switching the order of two consecutive non-constraining children of $x$.
        In order to show that there exists a $\varphi$-oriented interval representation that agrees with $T''$, it therefore suffices to show the following.
        \begin{claim*}
            Let $T^\star$ be the rotation of $T'$ we get by switching the order of two consecutive non-constraining children $y_i, y_{i+1}$ of $x$ in $T'$.
            There exists a $\varphi$-orientable interval representation $\I'$ that agrees with $T^\star$.
        \end{claim*}
        To prove the claim, observe that since~$\I'$ is normalized, for every non-constraining child $y_i$ of $x$, each interval of $\I'[B_i]$ is properly contained in every interval of $\I'[W_x \cup A_x^{T'}]$.
        Consider the induced representations $\I'[B_i]$ and $\I'[B_{i+1}]$ whose intervals are all
        properly contained in every interval of $\I'[W_x \cup A_x^{T'}]$.
        Let $I$ be the smallest interval that spans $\I'[B_i] \cup \I'[B_{i+1}]$.
        Observe that $I$ is properly contained in every interval of $\I'[W_x \cup A_x^{T'}]$ and every interval of $\I'[B_i] \cup \I'[B_{i+1}]$ is properly contained in $I$.
        We can therefore apply \cref{lem:rev-rep} twice---first to $\cont_{\I'}(I)$
        and then to $\cont_{\I''} (I_i)$ and $\cont_{\I''} (I_{i+1})$, respectively, where
        $\I'' = \mirr_I(\I')$---to obtain a $\varphi$-orientable interval
        representation~$\mathrm{mirr}_{I}(\I')$ from $\I'$ that agrees with $T^\star$, i.e., the
        order of $\I'[B_i]$ and $\I'[B_{i+1}]$ coincides with the order of $y_i$ and $y_{i+1}$ in $T^\star$; see \cref{fig:switch-p-children}.

        \begin{figure}
            \centering
            \textsfbf{(a)} \includegraphics[page=1,valign=t]{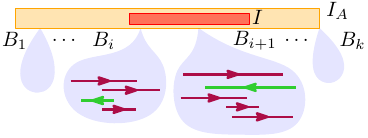}
             \hfil
            \textsfbf{(b)} \includegraphics[page=2,valign=t]{fig/switch-p-children}
            \caption{(a) Before and (b) after the swap of two consecutive unconstrained children using \cref{lem:rev-rep,lem:unconstrained-contained}.}
            \label{fig:switch-p-children}
        \end{figure}

        By repeatedly switching consecutive non-constraining children of $x$, we can therefore construct a $\varphi$-orientable interval representation that agrees with $T''$. This contradicts our choice of $T'$.

        Second, consider the case where $x$ is a Q-node.
        Since $\pi_x^T$ and $\pi_x^{T'}$ differ, it follows from \cref{lem:constrained-child} that $x$ is not constrained.
        Let $U \subseteq W_x$ be the set of vertices $u$ of $G$ that are assigned to $x$ and for which $S_u \neq S_{1,k}$.
        Observe that an interval representation of $G[U]$ is $\varphi$-orientable if and only if it is $\mathrm{rev}_U(\varphi)$-orientable:
        Since $x$ is not constrained, there are no horizontal (H) arcs between two vertices of $U$. 
        Hence, all adjacencies are realized either by overlapping intervals with opposite orientation or by containment.
        For both cases, it is easy to see that by reversing the orientation of the intervals, the edges remain correctly represented; see \cref{fig:q-node-rev-orient}.
        
        \begin{figure}
            \centering
            \textsfbf{(a)} \ \includegraphics[page=1,valign=t]{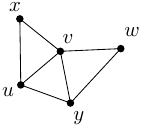} \hfil
            \textsfbf{(b)} \ \includegraphics[page=2,valign=t]{fig/q-node-rev-orient} \hfil
            \textsfbf{(c)} \ \includegraphics[page=3,valign=t]{fig/q-node-rev-orient}
            \caption{(a) A mixed graph $G$ containing only edges whose MPQ-tree is a single Q-node. (b)~and (c) show two orientations of the same interval representation of $U(G)$ where one is a reversal of the other.}
            \label{fig:q-node-rev-orient}
        \end{figure}

        In order to be able to apply~\cref{lem:rev-orient} to $B_x^T \cup U$, we argue that there are only (undirected) edges between $B_x^T \cup U$ and $A_x^T \cup (W_x \setminus U)$.
        Since $x$ is not constrained, there are no (V1) arcs between $B_x^T$ and $A_x^T$
        no (V2) arcs between $U$ and $A_x^T$, no (V3) arcs between $B_x^T$ and $W_x$ and no (H) arcs between $U$ and $W_x$.
        Therefore, since $\I'$ is normalized, the edges between $B_x^T \cup U$ and $A_x^T \cup W_x \setminus U$ are represented by containment.
        Hence, we apply \cref{lem:rev-orient} to obtain a $\mathrm{rev}_{\varphi}$-orientable interval representation $\I''$.
        As observed before, the interval representation induced by $W_x$ remains $\varphi$-orientable, therefore $\I''$ agrees with $T''$, a contradiction to our choice of $T'$.
    \end{proof}

    We remark that the MPQ-tree $T$ and the constrained children of each node can be computed in polynomial time.  Therefore, in polynomial time we can decide whether a mixed graph admits a $\varphi$-consistent clique order and compute one in the positive case.  In fact, the following lemma shows that this can be achieved in linear time.

    \begin{restatable}{lemma}{lemmpqalg}
    \label{lem:mpq-tree-alg}
    \restateref{lem:mpq-tree-alg}
        Let $G$ be a mixed graph and let~$\varphi$ be an orientation
        of $G$.  There is a linear-time algorithm that computes a clique order of $G$ that is~$\varphi$-consistent if $G$ admits a $\varphi$-consistent clique
        order.
    \end{restatable}


     \begin{proof}
        First, we compute the MPQ-tree, which can be done in
        time linear in~$|V(G)|+|E(G)|$~\cite{Korte1989}.   
        To satisfy the properties of \cref{lem:constrained-child}, we need to identify the constrained nodes~$x$ in $T$.  To do this, we process~$T$ bottom-up.  During the processing, for each node~$x$, we make use of a doubly-linked list~$\arcs(x)$ that shall contain all arcs that have one endpoint in $B_x^T \cup W_x$ and the other in $A_x^T$. 
        We assume that, when processing a node~$x$ with children~$y_1,\dots,y_k$, we receive the doubly-linked lists $\arcs(y_1), \dots, \arcs(y_k)$ as input.  Using this information, we then decide whether $x$ is constrained or not and the list~$\arcs(x)$.  We remark that the lists~$\arcs(y_1),\dots,\arcs(y_k)$ are destroyed during this computation.  This is not an issue as only $\arcs(x)$ is needed subsequently.  Since an arc is relevant for two lists~$\arcs(x)$ and~$\arcs(x')$ if and only if one is an ancestor of the other and we process the tree in a bottom-up manner, each arc is contained in only a single list at any given time. We can thus implement these lists intrusively, i.e., the arc directly stores the pointers used to implement the doubly-linked list it is contained in.  This allows to remove an arc from the list it is contained in $O(1)$ time.
        
        We assume that each vertex of $G$ knows the node in $T$ to which it is assigned and vice versa.  Moreover, if $x$ is a Q-node, each vertex $v$ of~$W_x$ also stores the left end right endpoints of its segment~$S_v = S_{i,j}$.  As a preprocessing step, we label each node of $T$ according to a DFS ordering of $T$ (including timestamps of the first and last visit for each vertex) starting at the root.  Observe that this allows us to test in constant time whether an arc is horizontal or vertical and in the latter case to determine its upper and lower endpoint.

        Initially, we compute the list $\arcs(x)$ for each leaf $x$ of $T$ by starting with an empty list and adding for each vertex $v$ in $W_x$, all vertical arcs incident to $v$.  Thus, processing a leaf takes~$O(\sum_{v \in W_x} \deg(v))$ time.
        
        Let $x$ be an inner node of $T$ with children $y_1,\dots, y_{k}$. Assume that we have $\arcs(y_1), \dots, \arcs(y_k)$.
        We show how to determine whether $x$ is constrained, and consequently the order $\pi_x^T$ of the children of $x$ in $T$, as well as how to compute the list $\arcs(x)$.        
        
        We first consider the case that $x$ is a P-node.  In this case, each constraining child must be either the first or the last child of~$x$, depending on the direction of the arc defining the constraint and the orientations of its endpoints.  We thus compute two lists $F$ and~$L$ of children that must be the first and last children, respectively. 
        The only constraints can come from vertical arcs with the lower endpoint in some~$B_i$ and the upper endpoint in~$W_x \cup A_x^T$.  Such an arc must then be contained~$\arcs(y_i)$.  Conversely, any arc in~$\arcs(y_i)$ by definition imposes such a constraint.
        Thus for every $i \in \{1,\dots,k\}$ we check if $\arcs(y_i)$ is empty.
        If not, we take the first arc in~$\arcs(y_i)$ and, depending on the orientation of its endpoints and the direction of the arc, we add $i$ to the respective list $F$ or $L$, to signify that it needs to be the first or the last child of~$x$, respectively.  Note that we do not access any elements from~$\arcs(y_i)$ beyond the first one.  After processing all children in this way, we proceed as follows.
        
        If~$|F| \ge 2$ or~$|L| \ge 2$, there is no order $\pi_x^T$ such that there exists a $\varphi$-orientable interval representation that agrees with $T$ and we can output an arbitrary clique ordering of $G$.
        If $|F| \leq 1$ and $|L| \leq 1$, we reorder the children of~$x$ accordingly. Recall that we only ever access the first element of~$\arcs(y_i)$, if any. Therefore, it is still possible that the order~$\pi_x^T$ we determined does not satisfy the constraints of $x$, since there may be a child~$y_i$, whose list~$\arcs(y_i)$ contains both an arc that constrains~$y_i$ to be the first child and an arc that constrains it to be the last child.   Since we only consider the first arc in~$\arcs(y_i)$, the algorithm will not note this.  However, in this case, there is no order~$\pi_x^T$ such that there exists a $\varphi$-orientable interval representation that agrees with $T$ and the order we computed for~$x$ is as good as any other. 
       
        It remains to compute the list~$\arcs(x)$.  
        For that
        we first concatenate the lists~$\arcs(y_1),\dots,\arcs(y_k)$ and afterwards for each vertex~$v \in W_x$, we go over all its incident arcs and we remove all those for which~$v$ is the upper endpoint (from $\arcs(x)$) and we add all those for which~$v$ is the lower endpoint to~$\arcs(x)$.  The time to time to process~$x$ is therefore bounded by~$O(k + \sum_{v \in W_x} \deg(v))$.
      
        Finally, consider the case that $x$ is a Q-node.  In this case, we know that its order $\pi_x^T$ is uniquely determined as soon as we find an arc that renders $x$ constrained.  We can thus stop searching for such an arc immediately once we have found the first one.  Before we describe this in detail, observe that, in the end, we form~$\arcs(x)$ in the same way as for P-nodes.  In particular, all arcs that are contained in some~$\arcs(y_i)$ but whose upper endpoint is assigned to~$x$ will be removed in this step.

        To check whether $x$ is constrained we proceed as follows.  Recall that there are three ways in which vertical arcs can constrain a Q-node.   Each vertical arcs constraining~$x$ whose lower endpoint belong to~$B_x^T$ is necessarily contained in one of the lists~$\arcs(y_i)$.  We thus scan the lists~$\arcs(y_1),\dots,\arcs(y_k)$ one after the other.  Let~$y_i$ be the current child.  We process its list~$\arcs(y_i)$ one by one.  Let~$a$ be the current arc and let~$v$ denote its upper endpoint.
        If $v$ is not assigned to $x$, then it belongs to~$A_x^T$ and $a$ is a vertical (V1) arc, i.e.,~$y_i$ must be the first or last child of~$x$, which either fixes~$\pi_x^T$ or even shows that no suitable order of the children exists (if~$i \ne 1,k$).
        If~$v$ is assigned to~$x$, then $a$ is a (V3) arc, fixing the way in which the intervals representing the lower and upper endpoint have to overlap in any $\varphi$-orientable interval representation, so again, the orientation of~$x$ is fixed.
        As soon as we find an arc that constrains $x$ we orient the Q-node~$x$ accordingly, stop the processing of the lists and move on to the computation of~$\arcs(x)$.  Note that any arc we process in this way is subsequently removed when forming~$\arcs(x)$, so the amortized processing cost is~$O(1)$ per list and thus~$O(k)$ in total.
        To determine the vertical (V2) arcs that constrain~$x$, we traverse~$W_x$ and for each arc~$a$ incident to a vertex in~$W_x$, we check whether it satisfies the conditions of the (V2) arc constraint.  This can be done in constant time, since vertices assigned to Q-nodes store their segment.

        It remains to check for constrained children stemming from horizontal arcs.  To this end, we go through all arcs incident to vertices in~$W_x$ and for each such arc that is horizontal, we check whether the segments of its endpoints are not the same.  If they are not, we order the Q-node~$x$ correspondingly.  Observe that, in this way, if there is any arc constraining $x$, then we identify one.  This takes~$O(k + \sum_{v \in W_x} \deg(v))$ time.  Afterwards, we form the list~$\arcs(x)$ as described above in the same running time.  Therefore, the total  time for processing $x$ is in~$O(k + \sum_{v \in W_x} \deg(v))$.

        Eventually, after processing the whole tree in this way, we have reordered the children of each node such that the resulting clique order~$\sigma$ is $\varphi$-consistent if such an ordering exists.  It remains to bound the running time.

        As noted before, computing $T$ takes time $O(|V(G)|+|E(G)|)$ and the size of~$T$ is linear in~$|V(G)|$.  Moreover, the initial DFS visit can be performed in~$O(|V(T)|)$ time.  Recall that the processing time of each node~$x$ is in~$O(c_x + \sum_{v \in W_x} \deg(v))$, where~$c_x$ is the number of children of~$x$.  Note that~$\sum_{x \in T} c_x \le n$ and, since each vertex $v \in V(G)$ is assigned to exactly one node~$x$ of $T$, we have $\sum_{x \in T} \sum_{v \in W_x} \deg(v) = \sum_{v \in V(G)} \deg(v) = 2 \cdot (|\Arcs(G)| + |E(G)|)$.  Therefore, the total running time is in $O(|V(G)|+|E(G)|)$.
%
%

    \end{proof}

This result shows that we can construct an MPQ-tree that gives us a clique ordering for a $\varphi$-orientable interval representation of $G$ if and only if $G$ is a oriented interval graph.
    It remains to
    construct a $\varphi$-oriented interval representation from a given clique ordering.

 Recall the definition of augmented graph $G^+$ from \cref{subsec:constructing-phi} and let $C = \{c_1,\dots, c_k\}$ be the set of clique point vertices.
    Let $V^L = \{v \in V(G) \mid \varphi(v) = -1\}$,
    let $G^L = G^+[V^L \cup C]$, and let
    $\sigma^L = C_1^L,\dots, C_k^L$, where, for $i \in \{1,\dots,k\}$,
    $C_i^L = C_i^+ \cap V(G^L)$; see
    \cref{fig:L-R-split-overview-c}.  Define $V^R$, $G^R$,
    $C_1^R,\dots,C_k^R$, and $\sigma^R$ analogously; see
    \cref{fig:L-R-split-overview-e}.
    Let $\varphi^L$ and $\varphi^R$ denote the orientations
    of~$G^L$ and $G^R$ that assign to each vertex the
    orientation~$-1$ and~$+1$, respectively.
    We show the following statement.

    \whenappendix{\lemlrsplit*\label{lem:L-R-split*}}


    \begin{figure}[tb]
        \centering
        \begin{subcaptiongroup}
            
        \textsfbf{(a)} \quad \includegraphics[page=2,valign=t]{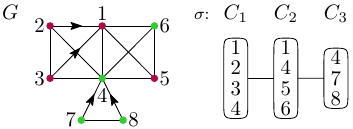}
        \phantomcaption\label{fig:L-R-split-overview-a}\hfil
        \textsfbf{(b)} \quad \includegraphics[page=5,valign=t]{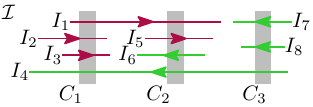}
        \phantomcaption\label{fig:L-R-split-overview-b}

\smallskip
{\color{lightgray}\hrule}
\smallskip
        
        \textsfbf{(c)} \quad \includegraphics[page=3,valign=t]{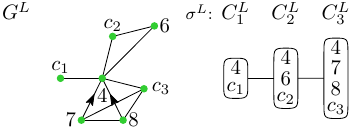}
        \phantomcaption\label{fig:L-R-split-overview-c}\hfil
        \textsfbf{(d)} \quad \includegraphics[page=6,valign=t]{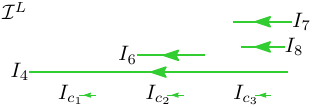}
        \phantomcaption\label{fig:L-R-split-overview-d}
        
\smallskip
{\color{lightgray}\hrule}
\smallskip

        \textsfbf{(e)} \quad \includegraphics[page=4,valign=t]{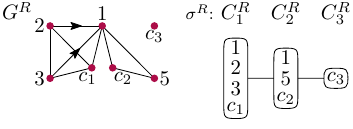}
        \phantomcaption\label{fig:L-R-split-overview-e}\hfil
        \textsfbf{(f)} \quad \includegraphics[page=7,valign=t]{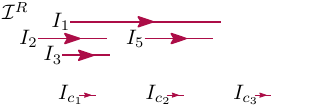}
        \phantomcaption\label{fig:L-R-split-overview-f}
        \end{subcaptiongroup}
      
        \caption{Example illustrating the proof of \cref{lem:L-R-split}.}
        \label{fig:L-R-split-overview}
    \end{figure}
      
    \begin{proof}
      ``$\Rightarrow$'': We take the $\varphi$-orientable interval
      representation~\I of~$G$ with clique order~$\sigma$ and split it
      into two representations, a representation~$\I^L$ for the
      left-oriented intervals and a representation~$\I^R$ for the
      right-oriented intervals. To both of these, we add, for each
      clique points in~$C$, a very short interval that intersects all
      intervals that form the corresponding clique. The first
      resulting representation is $\varphi^L$-oriented and has clique
      order~$\sigma^L$, the second is $\varphi^R$-oriented and has
      clique order~$\sigma^R$, as desired.

      ``$\Leftarrow$'': Let~$\I^L$ be the $\varphi^L$-orientable
      interval representation of~$G^L-C$, and let~$\I^R$ be the
      $\varphi^R$-orientable interval representation~$\I^R$ of~$G^R-C$.
      We join~$\I^L$ and~$\I^R$ as follows.  Since $\sigma^L$ and
      $\sigma^R$ both have a clique for each clique point vertex, we
      can simply set, for each $i \in \{1,\dots,k\}$,
      $C_i = (C_i^L \cup C_i^R) \setminus \{c_i\}$ and
      $\sigma = C_1,\dots,C_k$.  It remains to determine, for each
      $i \in \{1,\dots,k\}$, the order of the interval endpoints within~$C_i$;
      see \cref{fig:L-R-split-cliquepoint}.
      \begin{figure}[tb]
          \centering
          \includegraphics[page=8]{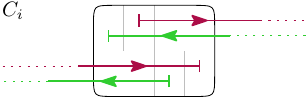}
          \caption{How to place the endpoints of left- and
            right-oriented intervals within a clique.}
          \label{fig:L-R-split-cliquepoint}
      \end{figure}
      To this end, we first
      place the left endpoints of the intervals in~$V(G)$ that start
      in~$C_i$ and extend to the right; then we place the right
      endpoints of the intervals in~$V(G)$ that start in~$C_i$ and
      extend to the left.  We order the left endpoints as follows: we
      start with the intervals in $V^L$ in the order in which they
      appear in~$\I^L$; then we add the intervals in $V^R$ in the
      order in which they appear in~$\I^L$.  We order the right
      endpoints accordingly. This fixes a set~\I of intervals
      that contains an interval for each vertex of~$G$;
      it can be computed in linear time by merging~$\I^L$ and~$\I^R$.

      Note that two intervals in~\I intersect if and only if there is
      at least one $i \in \{1,\dots,k\}$ such that both intervals contribute to
      clique~$C_i$. This ensures
      that all undirected edges of~$G$ are realized in~\I.  Every arc
      $(u,v)$ of $G$ is realized in~\I because either both~$u$ and~$v$ lie
      in $V^L \setminus C$ and $(u,v)$ is realized in~$\I^L$, or they
      both lie in $V^R \setminus C$ and $(u,v)$ is realized in~$\I^R$.
      Since the order of the endpoints of intervals in
      $V^L \setminus C$ (in $V^R \setminus C$) is the same in~$\I^L$
      (in~$\I^R$) and in $\I$, the arc $(u,v)$ is realized in~\I.
    \end{proof}

    We are now ready to prove the main result of this section.

    \thmgivenphigetrep*

    \begin{proof}
      We use \cref{lem:mpq-tree-alg} to compute, in linear time, a clique
      order~$\sigma$ such that $G$ admits a $\varphi$-orientable
      interval representation if and only if $G$ admits such a
      representation with clique order~$\sigma$. Then we split
      $(G,\sigma)$ into $(G^L,\sigma^L)$ and $(G^R,\sigma^R)$ in linear time.
      Since all vertices of~$G^L$ have the same orientation, the
      undirected edges of~$G^L$ are precisely the containment edges,
      and we can use \cref{cor:linear-time-given-sigma} to check
      in linear time whether $G^L$ admits a
      $\varphi^L$-orientable interval representation with clique
      order~$\sigma^L$.  Then we do the
      corresponding check with~$G^R$.  If both checks are successful,
      we use \cref{lem:L-R-split} to combine the
      resulting interval representations of~$G^L$ and~$G^R$ in linear time to a
      $\varphi$-oriented interval representation of~$G$.
    \end{proof}  
    
\end{toappendix}

\section{Concluding Remarks}
\label{sec:concluding}

In this work, we have initiated a systematic study of the interconnections between clique orders, orientations, and containment edges of oriented interval representations.
Given a mixed graph~$G$, we have characterized the orientations~$\varphi$ of~$G$ that, given a set~$\Econt \subseteq E(G)$, form a consistent pair~$(\varphi, \Econt)$ as well as the clique orders~$\sigma$ of $G$ that, given an orientation~$\varphi$ of~$G$, form a consistent pair~$(\varphi, \sigma)$.
We also presented linear-time algorithms for two constrained versions of the recognition problem, where, in addition to~$G$, either the set~$\Econt$ or the orientation~$\varphi$ is prescribed.
The problem of recognizing oriented interval graphs given a clique order~$\sigma$ of~$G$ seems to be a natural next step for future work.
The structural insights that we have gained in this paper could be a good starting point for tackling this problem, as well as the general recognition problem, both of which we leave open here.
We note, though, that even if a polynomial-time algorithm exists for the latter, it is not clear that it would also solve the constrained recognition problems that we have studied in this paper.

\bibliographystyle{plainurl}
\bibliography{references}

\end{document}